\documentclass[11pt]{article}
\usepackage{amsmath,amsfonts,amsthm,amssymb}
\usepackage{url}
\usepackage{graphicx}
\usepackage{caption} 
\usepackage[noend]{algpseudocode}
\usepackage{algorithmicx}
\usepackage{bbm}
\usepackage{float}
\usepackage{enumitem}
\usepackage{color}
\usepackage{mathtools}
\usepackage{comment}
\usepackage[a4paper, total={6.5in, 9in}]{geometry}

\newcommand\pto{\mathrel{\ooalign{\hfil$\mapstochar$\hfil\cr$\to$\cr}}}

\DeclareMathOperator*{\E}{\mathbb{E}}
\let\Pr\relax
\DeclareMathOperator*{\Pr}{\mathbb{P}}

\DeclareMathOperator*{\Null}{null}
\DeclareMathOperator*{\Var}{\mathrm{Var}}

\DeclareMathOperator*{\length}{length}

\DeclareMathOperator*{\poly}{poly}
\DeclareMathOperator*{\cost}{cost}

\newcommand{\cC}[0]{\mathcal{C}}
\newcommand{\cD}[0]{\mathcal{D}}

\newcommand{\be}[0]{\begin{enumerate}}
	\newcommand{\ee}[0]{\end{enumerate}}
\newcommand{\bi}[0]{\begin{itemize}}
	\newcommand{\ei}[0]{\end{itemize}}

\newcommand{\AlgComment}[1]{{~~\{#1\}}}
\newcommand{\ElseComment}[1]{\Else {~~\{#1\}}}

\renewcommand{\Null}{\textbf{null}}

\newcommand{\stp}[1]{\emph{#1}}

\newcommand{\Mb}{\overline{M}}

\newtheorem{theorem}{Theorem}
\newtheorem*{theorem*}{Theorem}
\newtheorem{lemma}[theorem]{Lemma}
\newtheorem{corollary}[theorem]{Corollary}

\newtheorem{proposition}[theorem]{Proposition}
\newtheorem{definition}[theorem]{Definition}

\title{Fast and Simple Edge-Coloring Algorithms}
\author{Corwin Sinnamon\thanks{Department of Computer Science, Princeton University, Princeton, NJ 08540, USA.\newline
\emph{Email:} {\tt sinncore@gmail.com}}}
\date{}
\begin{document}
\maketitle
\begin{abstract}
We develop sequential algorithms for constructing edge-colorings of graphs and multigraphs efficiently and using few colors.
Our primary focus is edge-coloring arbitrary simple graphs using $d+1$ colors, where $d$ is the largest vertex degree in the graph.
Vizing's Theorem states that every simple graph can be edge-colored using $d+1$ colors.
Although some graphs can be edge-colored using only $d$ colors, it is NP-hard to recognize graphs of this type [Holyer, 1981].
So using $d+1$ colors is a natural goal.
Efficient techniques for $(d+1)$-edge-coloring were developed by Gabow, Nishizeki, Kariv, Leven, and Terada in 1985, and independently by Arjomandi in 1982, leading to algorithms that run in $O(|E| \sqrt{|V| \log |V|})$ time.
They have remained the fastest known algorithms for this task.

We improve the runtime to $O(|E| \sqrt{|V|})$ with a small modification and careful analysis.
We then develop a randomized version of the algorithm that is much simpler to implement and has the same asymptotic runtime, with very high probability.
On the way to these results, we give a simple algorithm for $(2d-1)$-edge-coloring of multigraphs that runs in $O(|E|\log d)$ time.
Underlying these algorithms is a general edge-coloring strategy which may lend itself to further applications.
\end{abstract}

\section{Introduction}
\label{sec:intro}

An \emph{edge-coloring} of a graph is an assignment of colors to edges such that any two incident edges have different colors.
A \textit{$k$-edge-coloring} is one that uses at most $k$ colors.
Edge-coloring is a standard notion in mathematics and has been well studied in many contexts; this paper is about sequential algorithms for constructing them.

Edge-coloring graphs efficiently and in few colors is a classic problem in graph algorithms, with natural applications to various scheduling tasks, e.g.~\cite{AMSZ03, CY07, KN03, WHH97}.
This problem has been considered in the context of simple graphs~\cite{GNKLT85, Arj82}, multigraphs~\cite{HNS86, SS08}, bipartite multigraphs~\cite{Alon03, COS01, GK82}, and planar graphs~\cite{CN90, CK08}, to name a few.
There has also been extensive work on distributed edge-coloring algorithms, e.g.~\cite{BE13, GDP08}. Some recent work has explored dynamic edge-coloring~\cite{BCHN18, BGW21, DHZ19}, where the coloring must be maintained as edges are inserted and deleted.
Across this area of study, there is a natural tradeoff between the number colors used and the efficiency or simplicity of the algorithm. In most cases, using more colors yields simpler, faster algorithms.
In general, we prefer to use as few colors as possible.\footnote{The \textit{chromatic index} of $G$ (denoted $\chi'(G)$) is the smallest. number of colors in any edge-coloring of $G$.}

Let $G$ be a simple graph having $n$ vertices and $m$ edges, and let $d$ be the maximum degree of any vertex in $G$.
Clearly, at least $d$ colors are needed to edge-color $G$ because the edges incident to any vertex must be assigned different colors.
Vizing's Theorem~\cite{Viz64} states that $G$ is always $d+1$-edge-colorable.\footnote{Note that Vizing's theorem is only true for simple graphs; multigraphs can require up to $\lfloor 3d/2 \rfloor$ colors~\cite{Sha49}.}
This leaves a very small gap: $G$ requires either $d$ or $d+1$ colors.
Unfortunately, distinguishing between graphs that require $d$ colors and those that require $d+1$ colors is NP-hard~\cite{Hol81}.
Since it seems infeasible to always use the optimal number of colors is infeasible, so we settle for $d+1$.

It is then natural to aim for using $d+1$ colors.
A 1985 technical report of Gabow, Nishizeki, Kariv, Leven, and Terada~\cite{GNKLT85} presents an algorithm for edge-coloring simple graphs using $d+1$ colors.\footnote{See also~\cite{Arj82} and the discussion of that work in~\cite{GNKLT85}.}
Their algorithm runs in $O(m \sqrt{n \log n})$ time.
Although there has been significant progress for algorithms on some restricted graph classes (bipartite graphs, for example~\cite{COS01}), there have been no faster algorithms proposed to edge-color arbitrary simple graphs in $d+1$ colors.
This paper improves the runtime to $O(m \sqrt n)$, first by a deterministic algorithm, and then by a much simpler randomized one.

\subsection*{Contributions}

In this paper, we present three algorithms for edge-coloring graphs: \textproc{Greedy-Euler-Color}, \textproc{Euler-Color}, and \textproc{Random-Euler-Color}.
They are founded on a general recursive strategy for edge-coloring that we lay out in detail. The strategy is a generalization of the methods used by Gabow et al. in~\cite{GNKLT85}, and may help in developing other edge-coloring methods.

\bi
\item
\textproc{Greedy-Euler-Color} finds a $(2d-1)$-edge-coloring of a multigraph in $O(m\log d)$ time.
It is simple and efficient, and perhaps the most natural possible application of our strategy.
Note that $O(m\log d)$ time for $(2d-1)$-edge-coloring can be achieved using a known dynamic edge-coloring algorithm~\cite{BCHN18}.

\item
\textproc{Euler-Color} is an algorithm that uses $d+1$ colors and runs in $O(m\sqrt n)$ time, improving on the result of~\cite{GNKLT85}.
Surprisingly, this is achieved with one subtle change to the algorithm of~\cite{GNKLT85} and some additional analysis which removes a factor of $\sqrt{\log n}$ from the runtime.

\item
\textproc{Random-Euler-Color} is a much simpler randomized version of \textproc{Euler-Color} that uses $d+1$ colors and runs in $O(m \sqrt n)$ time with probability $1 - e^{-\poly(n)}$.
\ei

We consider \textproc{Random-Euler-Color} the primary contribution of this paper, due to its efficiency and simplicity.
To approach the same runtime, both \textproc{Euler-Color} and the algorithm of~\cite{GNKLT85} rely on a complex subroutine (\textproc{Color-Many} in this paper and \textproc{Parallel-Color} in~\cite{GNKLT85}). \textproc{Random-Euler-Color} avoids that type of procedure using randomness, and runs just as quickly with very high probability.

We also present three coloring subroutines: \textproc{Color-One}, \textproc{Random-Color-One}, and \textproc{Color-Many}.
They are needed by \textproc{Euler-Color} and \textproc{Random-Euler-Color}.

The paper is organized as follows. Section~\ref{sec:euler} presents the edge-coloring strategy. Sections~\ref{sec:greedyeuler},~\ref{sec:eulereuler}, and~\ref{sec:randomeuler} introduce the three main algorithms of this paper and analyze them.
The rest of the paper tackles the subroutines.
Section~\ref{sec:defs} introduces the necessary definitions, Section~\ref{sec:one} presents \textproc{Color-One}, Section~\ref{sec:random} presents \textproc{Random-Color-One}, and Section~\ref{sec:many} presents \textproc{Color-Many}.

\section{Edge-coloring with Recursion}
\label{sec:euler}

In this section, we present our edge-coloring strategy, setting the stage for the main algorithms of this paper.
It can be approached with little preliminary discussion.

\subsection{Preliminaries}

Let $G=(V,E)$ be a undirected multigraph with $n$ vertices, $m$ edges, and maximum degree $d$ (that is, $d = \max_{v \in V} \deg(v)$).
As shorthand, we shall write $uv$ or $vu$ to denote the edge $\{u, v\}$.
We assume without loss of generality that $G$ has no isolated vertices (vertices with no incident edges), and so $m \geq n/2$.\footnote{Isolated nodes are irrelevant to edge-coloring, so this is fair assumption.}

Assume we maintain a partial edge-coloring of $G$.
An edge is \emph{colored} if it is assigned a color in the edge-coloring and otherwise it is \emph{uncolored}.
Let $\ell$ denote the number of uncolored edges in $G$.
We say a color $\alpha$ is \emph{missing} at vertex $v$ if no edge incident to $v$ is colored by $\alpha$.
Let $M(v)$ be the set of colors missing at $v$.

\subsection{Strategy}

All of the algorithms in this paper employ the same strategy that we now describe.
It is a generalization of a recursive method originally applied to bipartite graphs in~\cite{Gab76} and later adapted to general graphs~\cite{GNKLT85, Arj82}.
The strategy can be applied to both simple graphs and multigraphs.

The strategy uses a recursive divide-and-conquer approach.
It splits the original graph into two edge-disjoint subgraphs, recursively produces an edge-coloring of each one, and then stitches together the colorings at a reduced cost.
The parameter $d$ (the maximum degree of a vertex in the graph) is important: The subgraphs will have maximum degree roughly half that of the original graph, and consequently they can be edge-colored more efficiently and with fewer colors.

The strategy is flexible in the number of colors that can be used to edge-color the graph (the \emph{available} colors).
Let $\#_c(d)$ be the number of available colors.
This paper will only study $\#_c(d) = d+1$ and $\#_c = 2d-1$, and in particular we treat $\#_c$ as a function of $d$, but in principle $\#_c$ could be any function of the graph.

The strategy works in four steps: \stp{Partition}, \stp{Recurse}, \stp{Prune}, and \stp{Repair}.
The only freedom lies in the \stp{Repair} step; the other three steps will be the same in all of our algorithms.
\begin{description}
\item[\emph{Partition.}]
Partition the edges of the graph into two edge-disjoint subgraphs so that each subgraph has half of the edges ($\leq \lceil m/2 \rceil$) and the maximum degree of each subgraph is at most half the maximum degree of the original graph ($\leq \lceil d/2 \rceil$).
This can be done in $O(m)$ time by the method in Section~\ref{sec:partition}.

\item[\emph{Recurse.}]
Recurse to edge-color each subgraph using $\#_c(\lceil d/2 \rceil)$ colors.
We require that the two edge-colorings use different sets of colors; if necessary, relabel one of the edge-colorings so that no color is used in both subgraphs.

Combine the edge-colorings by taking their union.
As the subgraphs are edge-disjoint and their edge-colorings use disjoint color sets, the union is an edge-coloring of the original graph.


\item[\emph{Prune.}]
At this point, the edge-coloring uses up to $2\#_c(\lceil d/2 \rceil)$ colors, which may be more than the allowed $\#_c(d)$ colors. (For example, if $\#_c(d) = d+1$, the current coloring may use up to $d+3$ colors.)
We must eliminate these extra colors.

Let $t$ be number of extra colors.
Choose the $t$ \emph{least common} colors and uncolor all edges with those colors.
Now only $\#_c(d)$ colors are used, and the number of uncolored edges is at most $mt/(\#_c+t)$.

\item[\emph{Repair.}]
To complete the coloring, we must somehow assign colors to those few uncolored edges without using more than the available colors.
There is no prescribed way to do this step.
We will use the subroutines mentioned at the start of this section for this purpose.
\end{description}

\noindent Finally, if $d \leq 1$ then edge-color the graph using a single color. This serves as the base case of the recursion.
The template for this strategy is presented below in pseudocode.
\begin{algorithmic}[1]
\Procedure {Euler-Template}{$G$}
\Statex {\emph{Base Case:}}
\State {If $d \leq 1$, color every edge by the same color and return}
\Statex {\textbf{\stp{Partition}}}
\State {Decompose $G$ into subgraphs $G_1$ and $G_2$}
\Statex {\textbf{\stp{Recurse}}}
\State\AlgComment {Recursively edge-color $G_1$ and $G_2$ using different sets of colors}
\State {\textproc{Euler-Template}$(G_1)$}
\State {\textproc{Euler-Template}$(G_2)$}
\State\AlgComment {$G$ is edge-colored by $\leq 2 \#_c(\lceil d/2 \rceil)$ colors}
\Statex {\textbf{\stp{Prune}}}
\While {more than $\#_c$ colors are used}
\State {Choose the least common color $\gamma$}
\State {Uncolor all edges colored by $\gamma$}
\EndWhile
\State\AlgComment {At most $\#_c$ colors are used now and $\ell \leq mt/(\#_c+t)$}
\Statex {\textbf{\stp{Repair}}}
\State\AlgComment {Somehow color all the uncolored edges using $\#_c(d)$ colors}
\State {\hspace{60pt}$\vdots$}
\EndProcedure
\end{algorithmic}

Now that the strategy is stated, the \stp{Repair} step will be the only interesting component of our algorithms.
The other steps will not change.
As long as the \stp{Repair} step stays within $\#_c(d)$ colors and colors all the uncolored edges, the procedure must output a complete edge-coloring.

The potential of this strategy lies in the fact that only a small number of uncolored edges need to be colored at each step, and so much of the actual work of edge-coloring is done on small graphs with few edges and low maximum degree.
Our \stp{Repair} steps run much faster on these graphs.
The efficiency of the \stp{Repair} step will determine the efficiency of the whole algorithm, since the other non-recursive work done in each iteration can be executed in $O(m)$ time.

\subsection{Euler Partitions}
\label{sec:partition}

Before we proceed, let us summarize how a multigraph can be split into two edge-disjoint subgraphs as needed by the \emph{Partition} step.

An \emph{Euler partition} is a partition of the edges of a multigraph into a set of edge-disjoint tours\footnote{A \emph{tour} is a walk that does not repeat an edge.} such that every odd-degree vertex is the endpoint of exactly one tour, and no even-degree vertex is an endpoint of a tour.
Such a partition can be found greedily in $O(m)$ time, simply by removing maximal tours of the graph until no edges remain.
The edges of the graph can then be split between the subgraphs by traversing each tour and alternately assigning the edges to the two subgraphs.
This yields two subgraphs, each having half of the edges (either $\lceil m/2\rceil$ or $\lfloor m/2\rfloor$) and maximum degree at most $\lceil d/2 \rceil$.
The entire procedure takes $O(m)$ time.

\section{Greedy-Euler-Color}
\label{sec:greedyeuler}

Our first application of the strategy is a very simple algorithm that edge-colors a multigraph in $2d-1$ colors in $O(m\log d)$ time.
We state the \stp{Repair} step; the other steps are unchanged from the template.

The \stp{Repair} step uses the local coloring routine \textproc{Greedy-Color}, which takes as input an uncolored edge $uv$ and colors it.
\textproc{Greedy-Color} simply checks each of the $2d-1$ available colors to find one that is missing at both $u$ or $v$.
Such a color must exist by pigeonhole principle: There are at least $d$ missing colors at each endpoint and $2d-1$ colors in total.
It colors $uv$ by that color.
The result is a legal partial edge-coloring, i.e. no color is reused at any vertex.
This takes $O(d)$ time in the worst case.

\textproc{Greedy-Color} is stated here in pseudocode. Recall that $M(v)$ is the set of colors missing at vertex $v$.

\begin{algorithmic}[1]
\Procedure {Greedy-Color}{$uv$}
\For {each color $\alpha \in \{1, \dots, 2d-1\}$}
\If {$\alpha \in M(u)$ and $\alpha \in M(v)$}
\State {Color $uv$ by $\alpha$}
\State {\textbf{return}}
\EndIf
\EndFor
\EndProcedure
\end{algorithmic}

The \stp{Repair} step just applies \textproc{Greedy-Color} to every uncolored edge. The result is a $(2d-1)$-edge-coloring.

\begin{theorem}
\textproc{Greedy-Euler-Color} edge-colors $G$ by $2d-1$ colors in $O(m \log d)$ time.
\end{theorem}

The proof is straightforward: The \stp{Repair} step takes $O(m)$ time, and $d$ is roughly halved with each recursive call.
This yields a recursion tree of depth $O(\log d)$, where the total work done at each level in the tree is $O(m)$.
The details are left as an exercise.

\section{Euler-Color}
\label{sec:eulereuler}
Assume henceforth that $G$ is simple.
\textproc{Euler-Color} builds a $(d+1)$-edge-coloring of $G$ in $O(m\sqrt{n})$ time.
Here the \stp{Repair} step will use a combination of the subroutines \textproc{Color-One} and \textproc{Color-Many}.
Both subroutines increase the number of colored edges in a given partial edge-coloring (without using more than $d+1$ colors).
They are presented in Sections~\ref{sec:one} and~\ref{sec:many}, respectively.
\begin{itemize}
\item\textproc{Color-One} colors a given uncolored edge in $O(n)$ time.
\item\textproc{Color-Many} colors $\Omega(\ell/d)$ of the uncolored edges in $O(m)$ time.
(Recall that $\ell$ denotes the number of uncolored edges in the graph.)
\end{itemize}

Notice that \textproc{Color-Many} takes amortized $O(md/\ell)$ time for each edge it colors.
Thus, as long as $\ell = \Omega(md/n)$, \textproc{Color-Many} is more efficient than \textproc{Color-One} in terms of the cost per edge colored.
This observation suggests the following \stp{Repair} step:
Use \textproc{Color-Many} while $\ell \geq 2md/n$, and then switch to \textproc{Color-One} for the remainder.


\begin{algorithmic}[1]
\State {\textbf{\stp{Repair}} (Euler-Color)}
\State\AlgComment {At most $d+1$ colors are used and $\ell \leq \frac{2m}{d+3}$}
\While {$\ell \geq 2md/n$}
\State {\textproc{Color-Many}$()$}
\EndWhile
\While {$\ell > 0$}
\State {Choose an uncolored edge $e$}        
\State {\textproc{Color-One}$(e)$}
\EndWhile
\end{algorithmic}

As discussed in the introduction, \textproc{Euler-Color} arises as a small modification to the algorithm of the same name in~\cite{GNKLT85} that runs in $O(m\sqrt{n\log n})$ time.
So what is the difference between this procedure and the original \textproc{Euler-Color} of~\cite{GNKLT85}?
They use a similar recursive strategy and two subroutines similar to ours: \textproc{Recolor} (analogous to \textproc{Color-One}) and \textproc{Parallel-Color} (analogous to \textproc{Color-Many}).
We prefer our versions of the subroutines, but we could use \textproc{Recolor} and \textproc{Parallel-Color} instead.
The only significant difference is that \textproc{Parallel-Color} is not integrated into the recursive strategy.
Their \stp{Repair} step applies \textproc{Recolor} to all uncolored edges.
When $d$ gets sufficiently small ($d \leq \sqrt{n/\log n}$), the recursive strategy is abandoned and the graph is edge-colored using only \textproc{Parallel-Color}.
This difference increases the runtime slightly to $O(m\sqrt{n \log n})$.
We bring \textproc{Color-Many} into the \stp{Repair} step to further exploit the tradeoff between these two subroutines, yielding the improved time.

We now analyze the complexity of \textproc{Euler-Color}.

\begin{theorem}
\label{thm:euler}
\textproc{Euler-Color} edge-colors a simple graph using $d+1$ colors in $O(m \sqrt{n})$ time.
\end{theorem}
\begin{proof}
The correctness of \textproc{Euler-Color} is clear from the strategy, for the \stp{Repair} step does not terminate until the graph is completely colored and it cannot use more than $d+1$ colors during that step.

Let us analyze the runtime.
At the start of the \stp{Repair} step, there are at most $2m/(d+3)$ uncolored edges.
The \stp{Repair} step behaves differently depending on whether $d$ is much larger or smaller than $\sqrt{n}$.

If $d < \sqrt{n}$, then \textproc{Color-Many} is run until the number of uncolored edges is reduced from $\ell \leq 2m/(d+3)$ to $2md/n$.
Let us bound the number of repetitions:
A constant number of applications of \textproc{Color-Many} decreases $\ell$ by a factor of $(1 - 1/d)$.
Since $(1-1/d)^d < 1/e$,
\[\ell\left(1 - \frac{1}{d}\right)^{2d\ln(2\sqrt{n}/d)} < \left(\frac{2m}{d+3}\right) \left(\frac{1}{e}\right)^{2\ln(2\sqrt{n}/d)} = \left(\frac{2m}{d+3}\right) \left(\frac{d^2}4{n}\right) < \frac{2md}{n}\]
Thus it requires $O(d\log(2\sqrt{n}/d))$ repetitions, where each repetition takes $O(m)$ time.
Afterwards it colors the remaining $md/n$ edges in $O(md)$ time using \textproc{Color-One}.
Hence the \stp{Repair} step takes $O(md \log(2\sqrt{n}/d))$ time.

If $d \geq \sqrt{n}$, then \textproc{Color-Many} is not run at all since $\ell \leq 2m/(d+3) \leq 2md/n$.
In this case, the \stp{Repair} step takes $O(mn/d)$ time.

Now consider the recursion tree for \textproc{Euler-Color}($G$).
This is a binary tree whose root corresponds to $G$, and the two children of the root correspond to the subgraphs $G_1$ and $G_2$ constructed during the \stp{Partition} step on $G$.
The children of the nodes for $G_1$ and $G_2$ are defined recursively.
A leaf corresponds to a subgraph of maximum degree 1.
Identify each node with its corresponding subgraph.
For a subgraph $H$, let $m_H$ and $d_H$ respectively denote the number of edges and maximum degree of $H$.
Every subgraph contains $n$ nodes.

Define the \emph{cost} of $H$ to be
\[\cost(H) \coloneqq \begin{cases}
~~\quad m_Hn/d_H \qquad&\text{if $d_H \geq \sqrt{n}$}\\
m_H d_H\log(2\sqrt{n}/d_H) &\text{if $d_H < \sqrt{n}$}
\end{cases}\]
This cost captures (up to a constant) all the time spent on the subgraph $H$, according to our earlier analysis.
Thus, the total time for \textproc{Euler-Color} is proportional to the sum of costs of all subgraphs in the recursion tree.

Consider the sum of costs over \textit{high-degree} subgraphs: the subgraphs $H$ with $d_H \geq \sqrt{n}$.
(We can assume here that $d \geq \sqrt n$, for otherwise there would be no high-degree subgraphs.)
These subgraphs have depth at most $\log(2d/\sqrt{n})$, since the vertex degrees roughly halve between a subgraph and its child.
Moreover, the subgraphs at a given depth partition the edges of $G$, so the sum of their edge counts is exactly $m$.
Hence, the total cost of all high-degree subgraphs is at most \[\sum_{t=0}^{\log(2d/\sqrt{n})} \frac{mn}{d} = \frac{mn}{d}\log(2d/\sqrt{n}) \leq \frac{mn}{\sqrt{n}} \log(2\sqrt{n}/\sqrt{n}) = m\sqrt{n}\]

Now consider the total cost of the remaining \textit{low-degree} subgraphs, which have $d_H < \sqrt{n}$.
Fix such a subgraph $H$, and let $H_1$ and $H_2$ be its children.
Observe that
\[\cost(H_1), \cost(H_2) \leq \left(\frac{m_H+1}{2}\right)\left(\frac{d_H+1}{2}\right)\log(4\sqrt{n}/d) \approx \frac{m_Hd_H}{4} \log(2\sqrt{n}/d) \approx \cost(H)/4\]
That is, the cost shrinks by a factor of 4, roughly, from a parent to its child.
Even though it does not shrink by a factor of 4 exactly, it is not hard to show that it shrinks by \textit{at least} a factor of 3, for sufficiently large $d_H$.
Hence, cost decreases geometrically with depth, shrinking by a factor of at least 3 at each level (except possibly near the bottom of the tree, vertex degrees are less than some constant).

Why does this matter? Each subgraph $H$ has 2 children, 4 grandchildren, 8 great-grandchildren, etc. So $H$ has $2^t$ descendants at a depth $t$ levels lower than $H$.
But the cost of each descendant is at most $\frac{1}{3^t} \cost(H)$.
Hence the \textit{combined} cost of the descendants of $H$ decreases as a geometric series: $\sum_{t=0} \left(\frac{2}{3}\right)^t \cost(H) = O(\cost H)$.\footnote{A constant number of levels near the bottom of the tree (for which the $d$ value is constant) may not follow this geometric trend, but their total cost is negligible.}

Therefore, the cost of all low-degree subgraphs is only a constant factor larger than the cost of the ``maximal'' low-degree subgraphs (those whose parents are high-degree).

In the case that $d \geq \sqrt n$, there are multiple maximal subgraphs, and partition the edges of $G$.
Moreover, each maximal subgraph $H$ has $d_H \in [\sqrt{n}/2, \sqrt{n})$.
Hence, their combined cost is at most
$m \sqrt{n}\log(4\sqrt{n}/\sqrt{n}) = O(m\sqrt{n})$.

If $d < \sqrt n$, $G$ is the only such maximal subgraph, and so the total cost is $O(\cost(G)) = O(md\log(2\sqrt{n}/d))$. Observe that $md\log(2\sqrt{n}/d) \leq m\sqrt{n}$ for all $d <\sqrt{n}$. Hence the total cost is $O(m\sqrt{n})$.
\end{proof}

The trouble with \textproc{Euler-Color} is that \textproc{Color-Many} is complicated, much more so than the rest of the algorithm.
Thankfully we can simplify enormously by injecting a little randomness.

\section{Random-Euler-Color}
\label{sec:randomeuler}

We now present \textproc{Random-Euler-Color}, a randomized version of \textproc{Euler-Color}.
It also uses $d+1$ colors and runs in $O(m \sqrt n)$ time with high probability.

We use a subroutine \textproc{Random-Color-One}, a randomized extension to \textproc{Color-One}, that is presented in Section~\ref{sec:random}.\footnote{ \textproc{Random-Color-One} essentially applies \textproc{Color-One} to an uncolored edge \textit{chosen uniformly at random}.
It also chooses one other variable (internal to \textproc{Color-One}) uniformly at random. These small changes are enough to improve the expected run time.}
\bi
\item
\textproc{Random-Color-One} colors an uncolored edge in $O(md/\ell)$ \emph{expected} time, or $O(n)$ time in the worst case.
\ei

Notice that the expected time matches the $O(md/\ell)$ amortized time-per-edge-colored that \textproc{Color-Many} achieves.
It also has the same $O(n)$ worst-case time bound of \textproc{Color-One}, meaning that \textproc{Random-Color-One} is an effective drop-in replacement for \textit{both} \textproc{Color-One} and \textproc{Color-Many}.

The \stp{Repair} step just applies \textproc{Random-Color-One} until the edge-coloring is complete.

\begin{algorithmic}[1]
\State {\textbf{\stp{Repair}} (Random-Euler-Color)}
\State\AlgComment {At most $d+1$ colors are used and $\ell \leq \frac{2m}{d}$}
\While {$\ell > 0$}
\State {\textproc{Random-Color-One}$()$}
\EndWhile
\end{algorithmic}

Now we analyze the expected runtime --- but the hard work has already been done!
The \stp{Repair} step of \textproc{Euler-Color} used a combination of \textproc{Color-One} and \textproc{Color-Many} to color the $O(m/d)$ uncolored edges.
We notice that \textproc{Random-Color-One} performs as well in expectation as \textit{either} subroutine.
Hence, the \stp{Repair} step of \textproc{Random-Euler-Color} can be no slower than that of \textproc{Euler-Color}.
Indeed, direct analysis of the \stp{Repair} step leads to the same cost function that appeared in our analysis of \textproc{Euler-Color}:
\[\sum_{\ell=1}^{2m/d} \min\left(n, \frac{md}{\ell}\right) =
\begin{cases}
~~\quad O(mn/d) \qquad&\text{ if $d \geq \sqrt n$}\\
O(md\log(2\sqrt{n}/d)) &\text{ if $d < \sqrt n$}
\end{cases}\]
By Theorem~\ref{thm:euler}, we obtain a tight bound on the expected runtime of \textproc{Random-Euler-Color}.

\begin{lemma}
\label{lem:randomeuler}
\textproc{Random-Euler-Color} edge-colors a simple graph by $d+1$ colors in $O(m\sqrt{n})$ expected time.
\end{lemma}

Now that we have established the expected time, we can go further.
In fact, $O(m\sqrt n)$ time is achieved except with \emph{exponentially} small probability, and we can show this much stronger claim with a little work.

\begin{theorem}
\label{thm:highprob}
\textproc{Random-Euler-Color} runs in $O(m\sqrt n)$ time with probability $1 - e^{-\sqrt{m}}$.
\end{theorem}
\begin{proof}
Let $T(G)$ be the combined runtime of all calls to \textproc{Random-Color-One} during an execution of \textproc{Random-Euler-Color} on $G$.
All other operations in \textproc{Random-Euler-Color} are deterministic and run in $o(m\sqrt{n})$ time, so it suffices to study $T(G)$.

\textproc{Random-Color-One} is applied many times during \textproc{Random-Euler-Color} to many partially edge-colored subgraphs of $G$.
Call them $G_1, \dots, G_k$, where $G_i$ has $m_i$ edges, maximum degree $d_i$, and $\ell_i$ uncolored edges.
Note that these are not random variables. The graphs $G_1, \dots, G_k$ are determined uniquely by the input graph because the \stp{Partition} step is deterministic and depends only on the structure of the graph.
Although the partial edge-colorings are randomly altered by \textproc{Random-Color-One}, the values of $m_i$, $d_i$, and $\ell_i$ are predictable.

Let $R_i$ be a random variable for the runtime of \textproc{Random-Color-One} on $G_i$.
Since the $G_i$ are deterministically chosen, the $R_i$ are independent.
Then $T(G) = \sum_{i=1}^k R_i$.
Lemma~\ref{lem:randomeuler} states that $\E[T(G)] = \E\left[\sum_{i=1}^k R_k\right] \leq C m\sqrt{n}$, for a sufficiently large constant $C$.

Recall that \textproc{Random-Color-One} takes $O(n)$ time in the worst case, and so $R_i \leq Cn$ (we reuse the constant $C$ here, without loss of generality).
Hence $\Var[R_i] \leq \E[R_i^2] \leq Cn \E[R_i]$ for $i = 1, \dots, k$.
By independence, we have \[\Var[T(G)] \leq Cn \E[T(G)] \leq C^2mn^{3/2}\]

We now use the following bound of Bernstein, proved in~\cite{Ber24} and communicated in~\cite{Ben62}.
\begin{theorem*}[Bernstein~\cite{Ber24}]
Let $X_1,\dots, X_n$ be independent random variables, let $\mathbf{X} = \sum_{i=1}^n X_i$, and let $\sigma^2 = \Var[\mathbf{X}] = \sum_{i=1}^n \Var[X_i]$. Suppose that $|X_{i} - \E[X_i]| \leq M$ for all $i$. Then, for $t > 0$,
\[\Pr\left[\mathbf{X} \geq \E[\mathbf{X}] + t\sigma\right] \leq \exp\left(-\frac{\frac{1}{2}t^2}{1 + \frac{1}{3}\frac{M}{\sigma}t}\right)\]
\end{theorem*}

Apply this theorem to $T(G)$. Here $M = Cn$ and $\sigma^2 = \Var[T(G)] \leq C^2 mn^{3/2}$, and we set $t = 3m^{1/2}/n^{1/4}$. We have $t\sigma \leq 3Cm\sqrt{n}$.
Hence,
\[\Pr\left[T(G) \geq \E[T(G)] + 3Cm\sqrt{n} \right]
\leq \exp\left(-\frac{\frac{9}{2}m/\sqrt{n}}{1 + \frac{1}{3}\frac{Cn}{Cm^{1/2}n^{3/4}}\frac{3m^{1/2}}{n^{1/4}}}\right)
= \exp\left(-\frac{\frac{9}{2}m/\sqrt{n}}{2}\right) \leq e^{-\sqrt{m}}
\]
The last inequality holds because $m \geq n/2$, by the assumption that there are no isolated nodes in any subgraph.
Thus, the runtime of \textproc{Random-Euler-Color} is $O(m\sqrt{n})$ except with probability $e^{-\sqrt m}$.
\end{proof}

In light of its strongly concentrated runtime and the simplicity of \textproc{Color-One} compared to \textproc{Color-Many}, \textproc{Random-Euler-Color} is preferable to \textproc{Euler-Color} in almost any context.

This concludes our study of the three major algorithms of this paper.
It remains to present the subroutines used by \textproc{Euler-Color} and \textproc{Random-Euler-Color}.

\section{Colors, Fans, and Paths}
\label{sec:defs}

The rest of the paper is devoted to presenting the subroutines \textproc{Color-One}, \textproc{Color-Many}, and \textproc{Random-Color-One} that are used in Sections~\ref{sec:eulereuler} and~\ref{sec:randomeuler}.
This section provides some terminology and tools that we will need to describe them, including the data structures underlying our implementations.

\subsection{Colors}

Since the remainder of this paper is concerned with edge-coloring in $d+1$ colors, we shall assume the available colors are $1, 2, \dots, d+1$.
We also restrict ourselves to simple graphs.

Let $G=(V,E)$ be a simple undirected graph with maximum degree $d$.
Let $c \colon E \pto [d+1]$ be a partial edge-coloring of $G$.
An edge $e$ is said to be \emph{colored} if $c(e)$ is defined and \emph{uncolored} otherwise.
Recall that a color $\alpha$ is said to be missing at vertex $v$ if no edge incident to $v$ is colored by $\alpha$.
Define $M(v) = \{\alpha \in [d+1] \mid \text{$\alpha$ is missing at $v$}\}$ and $\Mb(v) = [d+1] \setminus M(v)$.
Observe that $M(v)$ is always nonempty as $v$ has at most $d$ neighbours and there are $d+1$ colors to choose from.


\subsection{Fans}

Our subroutines will use a construction called a c-fan (sometimes called a Vizing fan), illustrated in Figure~\ref{fig:cfan}.
\begin{definition}
\label{def:cfan}
A \emph{c-fan} $F$ is a sequence $(\alpha, v, x_0, x_1, \dots, x_k)$ such that
\bi
\item $\alpha$ is a color in $M(v)$,
\item $x_1, \dots, x_k$ are distinct neighbours of $v$,
\item $vx_0$ is uncolored,
\item $vx_i$ is colored for $i = 1, \dots, k$, and
\item $c(vx_i)$ is missing at $x_{i-1}$ for $i = 1, \dots, k$.
\ei
Vertex $v$ is called the \emph{center} of $F$, and $F$ is said to be \emph{centered} at $v$.
The other vertices ($x_0, \dots, x_k$) are called the \emph{leaves} of $F$.
\end{definition}

\begin{figure}[!ht]
\begin{center}
\includegraphics[scale=0.6]{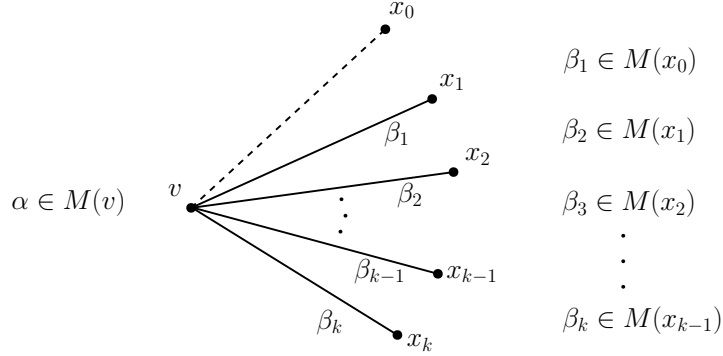}
\end{center}
\caption{A c-fan $(\alpha, v, x_0, x_1, \dots, x_k)$. Here $c(vx_i) = \beta_i$ and $\beta_i$ is always missing at $x_{i-1}$.}
\label{fig:cfan}
\end{figure}

The useful property of a c-fan is that we may ``rotate'' the colors of the fan without making the edge-coloring invalid.
Let $F = (\alpha, v, x_0, \dots, x_k)$ be a c-fan with $\beta_i = c(vx_i)$ for $i = 1, \dots, k$, and let $0 \leq j \leq k$.
To \emph{shift $F$ from $x_j$} means to set $c(vx_{i-1}) = \beta_i$ for $i = 1, \dots, j$, and to make $vx_j$ uncolored.
Since $\beta_i$ is required to be missing at $x_{i-1}$, the function $c$ is still a partial edge-coloring after a shift.
Note that $M(v)$ is unchanged and that $F' = (\alpha, v, x_j, x_{j+1}, \dots, x_k)$ is a c-fan after the shift.
Shifting a c-fan is shown in Figure~\ref{fig:fanshift}.

\begin{figure}[!ht]
\begin{center}
\includegraphics[scale=0.4]{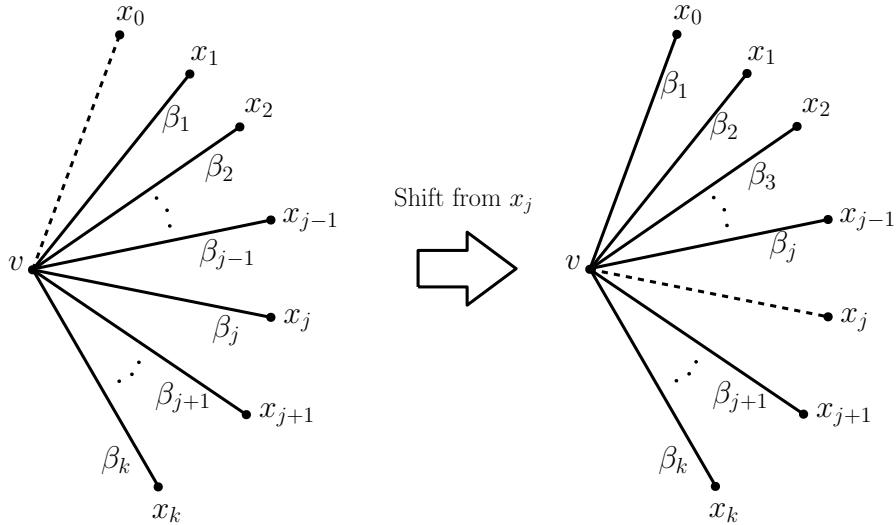}
\end{center}
\caption{A c-fan (left) and the resulting edge-coloring after shifting from $x_j$ (right). Solid lines are colored and dashed lines uncolored. A shift always preserves the legality of the edge-coloring.}
\label{fig:fanshift}
\end{figure}

Let us extend the definition of a c-fan by adding a parameter $\beta$.
Intuitively, a c-fan is \emph{primed} by $\beta$ once it is complete and does not need to be made larger.
A primed c-fan is exactly what we need to grow the edge-coloring, as we will see in the next sections.

\begin{definition}
\label{def:primedcfan}
A \emph{primed c-fan} is a c-fan $F = (\alpha, v, x_0, \dots, x_k)$ with an extra parameter $\beta$.
The color $\beta$ must satisfy one of the following conditions:
\be
\item $\beta \in M(x_k)$ and $\beta \in M(v)$, or
\item $\beta \in M(x_k)$ and there is a leaf $x_j$ such that $c(vx_j) = \beta$.
\ee
Setting $\beta$ is called \emph{priming} $F$, and we say $F$ is \emph{primed} by $\beta$.
\end{definition}

\subsection{Alternating Paths}
\label{sec:alternating}

We require one more tool to describe our algorithms.
For any two colors $\alpha$ and $\beta$, restricting $G$ to the edges colored by $\alpha$ or $\beta$ yields a subgraph of $G$ in which every component is either a path or an even-length cycle.
We call each component an $\alpha\beta$-path or an $\alpha\beta$-cycle, as appropriate.
We shall not require that $\alpha$ and $\beta$ be distinct, and an $\alpha\beta$-path may consist of a single vertex.
An $\alpha\beta$-path may also be called an \emph{alternating} path when $\alpha$ and $\beta$ are not specified.


To \emph{flip} an alternating path means to interchange the colors of its edges.
Any alternating path can be flipped in a partially edge-colored graph and the resulting assignment of colors will still be a partial edge-coloring.

Using the data structures in the next section, an alternating path can be flipped in time proportional to its length.

\subsection{Data Structures}
\label{sec:datastructs}

Our algorithms use the following elementary data structures. Here we diverge from~\cite{GNKLT85} by introducing a dictionary data structure that simplifies our subroutines. See Appendix~\ref{app:dictionary} for further discussion of this structure.

\bi
\item
We represent $G$ by storing for each vertex a list of its incident edges.
Each edge has a color field to track its color, if any.

\item
We keep a data structure $\mu(v)$ that stores colors missing at each vertex $v$.
However, $\mu(v)$ will not store all the colors in $M(v)$, as that could require too much time and space to initialize.
Instead $\mu(v)$ contains only the colors in $M(v) \cap [\deg(v)+1]$.
There must always be some color in $\mu(v)$ because $M(v) \cap [\deg(v)+1]$ is always nonempty.

We implement $\mu(v)$ as a doubly-linked list containing those colors in $M(v) \cap [\deg(v)+1]$ together with an array of length $\deg(v) + 1$.
The array contains pointers to the $\deg(v)+1$ nodes that can be in the list.
We can maintain $\mu(v)$ with constant overhead: Every time an edge incident to $v$ is colored by $\gamma \in [\deg(v) + 1]$, find the list node for $\gamma$ (using the array) and remove it from the list.
Every time an edge incident to $v$ with a color in $[\deg(v) + 1]$ is uncolored, append the node for that color to the linked list.
We can initialize $\mu(v)$ in $O(\deg(v))$ time.

\item
We also use a dictionary data structure $\cD$ that maps a vertex-color pair $(v, \gamma)$ to the edge incident to $v$ colored by $\gamma$, if there is such an edge.
It supports three operations that each take constant time: \textproc{Search}, \textproc{Insert}, and \textproc{Delete}.

$\cD$ can be initialized in $O(\sqrt{mnd})$ time, which is no more than $O(m\sqrt n)$.
Our implementation of $\cD$ is unusual because it must achieve this low initialization cost, but it operates by elementary methods.
In Appendix~\ref{app:dictionary}, we describe the implementation of $\cD$ and how to integrate $\cD$ into the recursive algorithms of Section~\ref{sec:euler}. The latter topic merits discussion because $\cD$ can only be initialized once, yet it must work for all the subgraphs in the recursion.
\ei

A data structure like $\cD$ was not present in~\cite{GNKLT85}.
It allows us to speed up or simplify all of our subroutines, and it is not known how to make a procedure like \textproc{Random-Color-One} function without such a data structure.

In summary, our data structures permit the following operations in constant time for any vertex $v$ and color $\gamma$: check whether $\gamma \in M(v)$; get a color in $M(v)$; color an edge; uncolor an edge; find the edge incident to $v$ colored by $\gamma$. As a corollary, we are able to shift a c-fan in time proportional to its size and flip an alternating path in time proportional to its length.

\section{Color-One}
\label{sec:one}

Our basic edge-coloring subroutine \textproc{Color-One} is analogous to \textproc{Recolor} in~\cite{GNKLT85} and derived from Vizing's original work~\cite{Viz64}.

\textproc{Color-One} first chooses an uncolored edge $vx_0$ and a color $\alpha \in M(v)$.
It calls a procedure \textproc{Make-Primed-Fan}, which builds a fan $F = (\alpha, v, x_0, x_1, \dots, x_k)$ primed by a color $\beta \in M(x_k)$.
It then \emph{activates} the fan by calling \textproc{Activate-c-Fan}, which uses the fan $F$ to color $vx_0$.
\begin{algorithmic}[1]
\Procedure {Color-One}{{}}
\State {Choose an uncolored edge $vx_0$}
\State {Choose any $\alpha \in M(v)$}
\State {$F \gets \textproc{Make-Primed-Fan}(v, x_0, \alpha)$}
\State {$\textproc{Activate-c-Fan}(F)$}
\EndProcedure
\end{algorithmic}

We describe \textproc{Make-Primed-Fan} and \textproc{Activate-c-Fan} in turn.\footnote{In most presentations of this type of coloring procedure, the construction of the c-fan and the activation of the fan are not treated as separate components. This bulky presentation will help us when we state \textproc{Random-Color-One} and \textproc{Color-Many}.}

\subsection*{Make-Primed-Fan}
\textproc{Make-Primed-Fan} creates a primed c-fan around $vx_0$, as follows.
Initialize a fan $F = (\alpha, v, x_0)$. We add leaves to $F$ until we find an opportunity to prime $F$.
Leaves are added greedily.
Pick a color $\beta$ missing at $x_k$, the last leaf of $F$.
If $\beta$ is missing at $v$, then we can prime $F$ by $\beta$; this satisfies the first way for $F$ to be primed.
Otherwise, there is a vertex $x_{k+1}$ with $vx_{k+1}$ colored by $\beta$.
If $x_{k+1}$ is already a leaf of $F$, then we can again primes $F$ by $\beta$; this satisfies the second way for $F$ to be primed.
If $x_{k+1}$ does not already appear in $F$, then appended to $F$ as the last leaf.
Return $F$ as soon as it is primed.

\begin{algorithmic}[1]
\Procedure {Make-Primed-Fan}{$v, x_0, \alpha$}
\Require {$vx_0$ uncolored}
\Require {$\alpha \in M(v)$}
\State {$F \gets (\alpha, v, x_0)$}
\State {$k \gets 0$}
\While {$F$ is not primed}
\State {Pick any $\beta \in M(x_k)$}
\If {$\beta \in M(v)$}
\State {Prime $F$ by $\beta$}
\Else \AlgComment {$\beta \not\in M(v)$}
\State {Find neighbour $x_{k+1}$ such that $c(vx_{k+1}) = \beta$}
\If {$x_{k+1} \in \{x_1, \dots, x_k\}$}
\State {Prime $F$ by $\beta$}
\Else\AlgComment {$x_{k+1}$ is not already in $F$}
\State {Append $x_{k+1}$ to $F$}
\State {$k \gets k + 1$}
\EndIf
\EndIf
\EndWhile
\State\Return {$F$}
\EndProcedure
\end{algorithmic}

\begin{lemma}
\label{lem:makefan}
\textproc{Make-Primed-Fan} returns a primed c-fan in $O(\deg(v))$ time.
\end{lemma}
\begin{proof}
By construction, $F$ is a c-fan at all times during \textproc{Make-Primed-Fan}.
Each iteration of the loop adds a new leaf to $F$ or primes $F$, and no leaf can be added twice; thus the loop repeats at most $\deg(v)$ times before $F$ is primed and returned.

Each operation within the loop takes constant time --- this is clear for every line except line 10, where we check whether $x_{k+1}$ appears already as a leaf of $F$. We can make this operation take constant time by marking each leaf when it is added to $F$ and unmarking them all at the end.
Hence the procedure takes $O(\deg(v))$ time.
\end{proof}

\subsection*{Activate-c-Fan}
\textproc{Activate-c-Fan} accepts a c-fan $F = (\alpha, v, x_0, x_1, \dots, x_k)$ primed by $\beta$ and uses $F$ to extend the coloring.
By Definition~\ref{def:primedcfan}, there are two ways for $F$ to be primed by $\beta$: Either $\beta \in M(v)$ or there is a leaf $x_j$ with $c(vx_j) = \beta$ (recall that $\beta \in M(x_k)$ in both cases).
\textproc{Activate-c-Fan} handles them differently.

If $\beta \in M(v)$, then shift $F$ from $x_k$ (so that $vx_k$ becomes uncolored) and color $vx_k$ by $\beta$. This extends the coloring by one edge, and we return.

Otherwise, $\beta$ is not missing at $v$, but there is some leaf $x_j$ where $vx_j$ is colored by $\beta$.
Observe that $\beta$ is missing at $x_k$ and at $x_{j-1}$, by construction of the c-fan.

Consider the $\alpha\beta$-path $P$ beginning at $v$.
Recall that $\alpha$ is missing at $v$ and $vx_j$ is colored by $\beta$.
Hence, $v$ is an endpoint of $P$, with $vx_j$ being the first edge of $P$.
Path $P$ ends at some other vertex $w$ (note that $w$ is missing exactly one of $\alpha$ or $\beta$).
Flip $P$ (i.e.\ interchange $\alpha$ and $\beta$ on its edges).
After the flip, $vx_j$ is colored by $\alpha$, and $\beta$ is missing at $v$.
The flip also affects $w$: If $\alpha$ was missing at $w$, now $\beta$ is missing, and vice versa.
The set of missing colors does not change for any vertices besides $v$ and $w$.

We now act in two cases depending on $w$. The cases are illustrated in Figure~\ref{fig:colorone}.
\be[label={\textbf{Case \Roman*:}}, leftmargin=\widthof{\textbf{Case II:}}+\labelsep]
\item If $w \not= x_{j-1}$, then $\beta$ is still missing at $x_{j-1}$. Observe that $F' = (\beta, v, x_0, \dots, x_{j-1})$ is a c-fan (this holds even if $w$ is a leaf of $F'$). Shift $F'$ from $x_{j-1}$ and color $vx_{j-1}$ by $\beta$.
\item If $w = x_{j-1}$, then $\beta$ is no longer missing at $x_{j-1}$ (now $\alpha$ is missing instead). Since $vx_j$ is colored by $\alpha$ and $\beta$ is missing at $v$, the sequence $F' = (\beta, v, x_0, x_1, \dots, x_k)$ is a c-fan. Moreover, $\beta$ must still be missing at $x_k$ because $w \not= x_k$. Shift $F'$ from $x_{k}$ and color $vx_{k}$ by $\beta$.
\ee

In either case, the procedure succeeds in extending the edge-coloring by $vx_0$ without uncoloring any other vertices.

\textproc{Activate-c-Fan} is given in pseudocode below.

\begin{algorithmic}[1]
\Procedure {Activate-c-Fan}{{$F$}}
\Require {$F = (\alpha, v, x_0, \dots, x_k)$ is a c-fan primed by $\beta$}
\If {$\beta \in M(v)$}
\State\AlgComment {$\beta \in M(v)$ and $\beta \in M(x_k)$}
\State {Shift $F$ from $x_k$}
\State {Color $vx_k$ by $\beta$}
\Else\AlgComment {$\beta \not\in M(v)$}
\State {Let $x_j$ be the leaf of $F$ with $c(vx_j) = \beta$}
\State\AlgComment {$\beta \in M(x_{j-1})$ and $\beta \in M(x_k)$}
\State {Flip the $\alpha\beta$-path $P$ beginning at $v$}
\State {Let $w$ be the other endpoint of $P$}
\If {$w \not= x_{j-1}$}
\State\AlgComment {$\beta \in M(x_{j-1})$}
\State {Shift $F$ from $x_{j-1}$}
\State {Color $vx_{j-1}$ by $\beta$}
\Else
\State\AlgComment {$\alpha \in M(x_{j-1})$}
\State {Shift $F$ from $x_k$}
\State {Color $vx_k$ by $\beta$}
\EndIf
\EndIf
\EndProcedure
\end{algorithmic}

\begin{figure}[ht!]
\begin{center}
\includegraphics[scale=0.46]{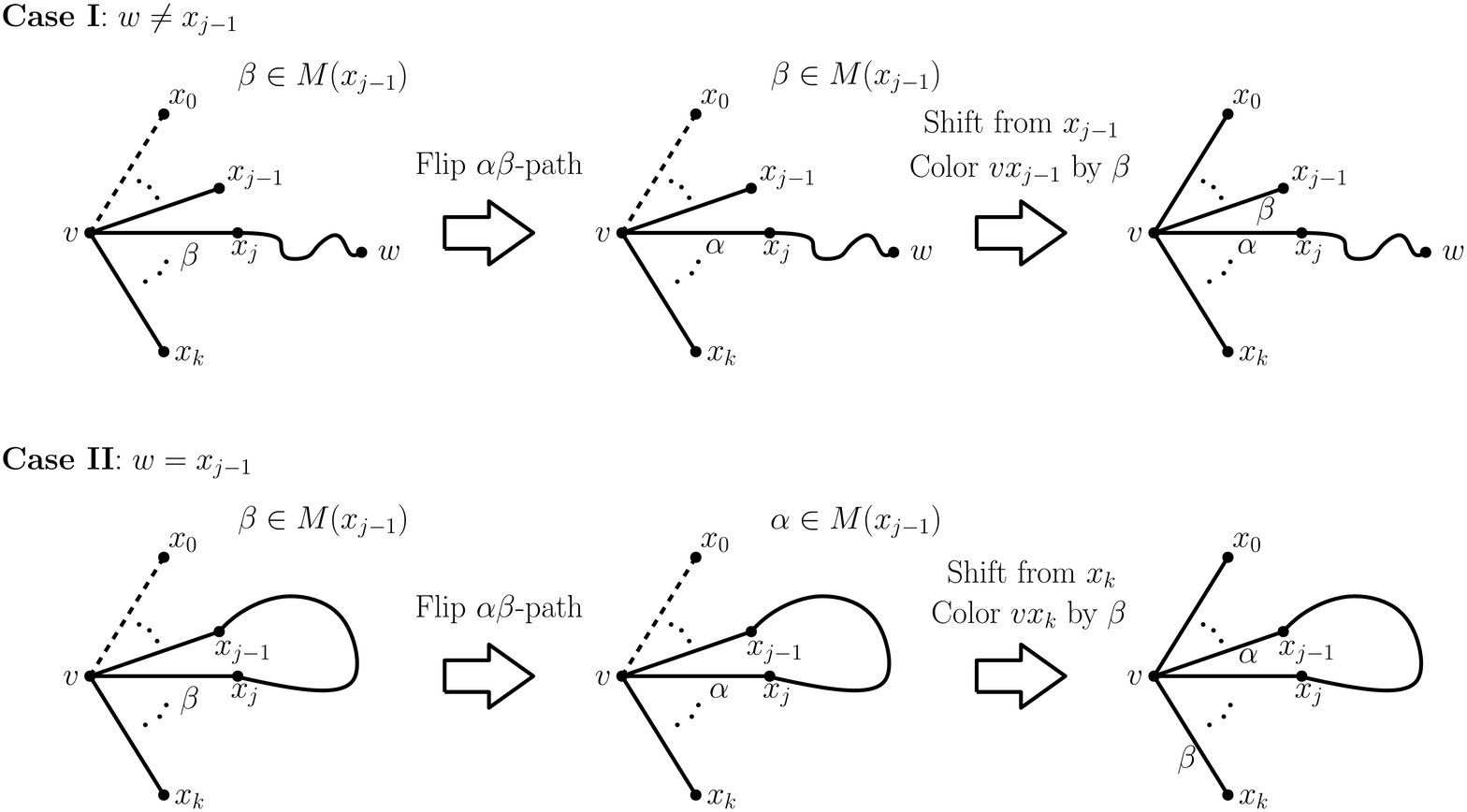}
\end{center}
\caption{The two cases of flipping the $\alpha\beta$-path in \textproc{Activate-c-Fan}.}
\label{fig:colorone}
\end{figure}

For convenience, let us assume going forward that \textproc{Activate-c-Fan} always flips one $\alpha\beta$-path with an endpoint at $v$.
This is indeed the case if \textproc{Activate-c-Fan} reaches line 9.
If it instead executes lines 3---5, we can just pretend that the $\alpha\beta$-path beginning at $v$ is flipped; it does no harm to assume this because the path has length 0.

\begin{lemma}
\label{lem:activatecfan}
\textproc{Activate-c-Fan} colors one edge in $O(\deg(v) + \length(P))$ time, where $P$ is the length of the flipped path.
\end{lemma}

\begin{proof}
Case analysis of \textproc{Activate-c-Fan} shows that it outputs a legal edge-coloring. Observe that no previously-colored edges are uncolored by \textproc{Activate-c-Fan}, and the first edge of the fan becomes colored. Hence it increases the number of colored edges by one.

The runtime is clear, as \textproc{Activate-c-Fan} takes $O(\deg(v))$ time for all operations except for flipping the path, which takes $O(\length(P))$ time.
\end{proof}

Lemmas~\ref{lem:makefan} and~\ref{lem:activatecfan} complete the analysis of \textproc{Color-One}.
Theorem~\ref{thm:dettime} shows that \textproc{Color-One} matches the behaviour promised in Section~\ref{sec:eulereuler}.

\begin{theorem}
\label{thm:dettime}
\textproc{Color-One} colors an uncolored edge $vx_0$ in $O(\deg(v) + \length(P))$ time in the worst case.
\end{theorem}

\section{Random-Color-One}
\label{sec:random}

While \textproc{Color-One} can take $\Theta(n)$ time when the flipped path is very long, we might expect better behaviour ``on average''.
Indeed, it cannot happen that every alternating path induced by an edge-coloring is very long.
It turns out that randomization helps a great deal in expectation.

The intuition is this: Let \textproc{Color-One} choose $vx_0$ at random from the uncolored edges.
\textproc{Color-One} can only flip an alternating path if $v$ is an endpoint of the path.
Since there are many different edges $vx_0$ to choose from, any particular alternating path is unlikely to be chosen.

This intuition is not far off, but it glosses over an important fact: If a vertex $v$ is incident to many uncolored edges, it is much more likely to be chosen for \textproc{Color-One}. If $v$ happens to be the endpoint of a very long path, \textproc{Color-One} may be slow in expectation.
To compensate for this, we also choose $\alpha$ at random. The flipped path must use $\alpha$ as one of its colors, so more choices for $\alpha$ translates to more options for the flipped path. Moreover, if $v$ is incident to many uncolored edges, it must also be missing many colors.

In summary, run \textproc{Color-One} with $vx_0$ and $\alpha$ chosen randomly.
We call this randomized procedure \textproc{Random-Color-One}.

\begin{algorithmic}[1]
\Procedure {Random-Color-One}{{}}
\State {Choose an uncolored edge $vx_0$ uniformly at random}
\State {Choose $\alpha \in M(v)$ uniformly at random}
\State {$F \gets \textproc{Make-Primed-Fan}(v, x_0, \alpha)$}
\State {$\textproc{Activate-c-Fan}(F)$}
\EndProcedure
\end{algorithmic}

\begin{lemma}
\label{lem:expected}
\textproc{Random-Color-One} colors an uncolored edge. It takes $O(md/\ell)$ time in expectation and $O(n)$ time in the worst case.
\end{lemma}
\begin{proof}
The correctness of \textproc{Random-Color-One} follows from Lemmas~\ref{lem:makefan} and~\ref{lem:activatecfan}, as does the worst-case $O(n)$ runtime.
It remains to show the expected time.

\textproc{Random-Color-One} chooses $v$, $x_0$, and $\alpha$ by a random process.
Treat these as random variables and let $P$ be a random variable representing the flipped path.
Since \textproc{Random-Color-One} runs in $O(d + \length(P))$ time, it is sufficient to show that $\E[\length(P)] = O(md/\ell)$.

Let $Q$ be an alternating path in the graph with endpoints $x$ and $y$.
Since they are endpoints, $x$ and $y$ must each be missing exactly one of the colors that appear on $Q$.
Say $\gamma_x \in M(x)$ and $\gamma_y \in M(y)$ are both colors used in $Q$.
In order for $P$ to take the value $Q$, we need either $v = x$ and $\alpha = \gamma_x$ or $v = y$ and $\alpha = \gamma_y$.

Since $vx_0$ is chosen uniformly at random from the uncolored edges, we have \[\Pr[v = x] = \frac{\deg(x) - |\Mb(x)|}{2\ell} \leq \frac{|M(x)|}{2\ell}\]
and similarly $\Pr[v = y] \leq |M(y)|/2\ell$. Since $\alpha$ is chosen uniformly at random from $M(v)$,
\begin{align*}
\Pr[P = Q] &\leq \Pr[v = x \text{ and } \alpha = \gamma_x]~+~\Pr[v = y \text{ and } \alpha = \gamma_y]\\
&= \Pr[v = x] \Pr[\text{$\alpha = \gamma_x$} \mid v = x]~+~\Pr[v = y] \Pr[\text{$\alpha = \gamma_y$} \mid v = y]\\
&\leq \frac{|M(x)|}{2\ell}\cdot\frac{1}{|M(x)|}~+~\frac{|M(y)|}{2\ell}\cdot\frac{1}{|M(y)|}\\
&= \frac{1}{\ell}
\end{align*}

Now, letting $Q$ range over all alternating paths in $G$, we have
\[\E[\length(P)] = \sum_{Q} \length(Q) \cdot \Pr[P = Q] \leq \sum_{Q} \frac{\length(Q)}{\ell}\]

Observe that every colored edge in $G$ contributes to at most $d+1$ different alternating paths in $G$. 
Hence $\sum_{Q} \length(Q) \leq m(d+1)$, and so $\E[\length(P)] \leq m(d+1)/\ell$.
\end{proof}

Thus, with this trivial modification to \textproc{Color-One}, the runtime is much improved.
It is especially significant because it does as well (in expectation) as the better of \textproc{Color-One} and \textproc{Color-Many} in terms of the time per edge colored, yet it is no more difficult to implement than \textproc{Color-One}.
As we showed via \textproc{Random-Euler-Color} in Section~\ref{sec:randomeuler}, \textproc{Random-Color-One} is all we need to edge-color a graph in $O(m\sqrt n)$ time with high probability.

\section{Color-Many}
\label{sec:many}

This section presents \textproc{Color-Many}, a subroutine that colors $\Omega(\ell/d)$ uncolored edges in a graph in $O(m)$ time.
In other words, it colors edges at an average cost of $O(md/\ell)$ time per edge.
In that sense, \textproc{Color-Many} serves as a deterministic replacement for \textproc{Random-Color-One}.
We used \textproc{Color-Many} as a subroutine in \textproc{Euler-Color} to achieve a runtime of $O(m\sqrt n)$.

\textproc{Color-Many} is our version of \textproc{Parallel-Color}, a procedure presented in~\cite{GNKLT85}.
As far as our results are concerned, the differences between \textproc{Color-Many} and \textproc{Parallel-Color} are unimportant---they have the same asymptotic efficiency per edge that they color.
\textproc{Color-Many} uses the essential ideas and mechanisms of \textproc{Parallel-Color}, and indeed they are similar in many ways.
We offer a more complete and thorough analysis of the procedure and we introduce some new ideas in \textproc{Color-Many} to simplify the presentation and implementation.
Although \textproc{Parallel-Color} could be used in place of \textproc{Color-Many}, we believe that \textproc{Color-Many} is preferable.

Despite these improvements, \textproc{Color-Many} is technical and requires some care.
Before we begin, we need some additional machinery.

\subsection*{Coloring with u-Fans}
We need a second kind of fan called a \emph{u-fan}.
Like a c-fan, a u-fan can be used to color an uncolored edge by a procedure we call \textproc{Activate-u-Fan}.

\begin{definition}
\label{def:ufan}
A \emph{u-fan} $F$ is a sequence $(\alpha, v, x_1, \dots, x_k)$ where $k \geq 2$, such that 
\bi
\item $x_1, \dots, x_k$ are distinct vertices that are all incident to $v$,
\item $vx_i$ is uncolored for $i = 1, \dots k$, and
\item $\alpha$ is a color that is missing at every $x_i$ and not missing at $v$; that is, \[\alpha \in \Mb(v) \cap \bigcap_{i=1}^{k} M(x_i).\]
\ei
Vertex $v$ is called the \emph{center} of $F$, and $F$ is said to be \emph{centered} at $v$.
The vertices $x_1, \dots, x_k$ are called the \emph{leaves} of $F$.
\end{definition}

Note that we require a u-fan to have at least two leaves.
We need this property to reliably color an uncolored edge using a u-fan.
If a u-fan is reduced to have less than two leaves, we call it \emph{degenerate}.



\textproc{Activate-u-Fan} takes a u-fan $F = (\alpha, v, x_1, \dots, x_k)$ and a color $\beta$ that is missing at $v$.
It flips the $\alpha\beta$-path beginning at $v$ so that $\alpha$ is missing at $v$.
If the flipped path ends at one of the leaves, then $\alpha$ is no longer missing there, so that leaf is removed from $F$.
A leaf $x$ of $F$ is chosen and $vx$ is colored by $\alpha$.
Then $x$ is removed from $F$ so that the resulting fan is still a u-fan (possibly a degenerate one).

\begin{algorithmic}[1]
\Procedure {Activate-u-Fan}{{$F, \beta$}}
\Require {$F = (\alpha, v, x_1, \dots, x_k)$ is a u-fan}
\Require {$\beta \in M(v)$}
\State {Flip the $\alpha\beta$-path $P$ beginning at $v$}
\State {Let $w$ be the endpoint of $P$}
\If {$w$ is a leaf of $F$}
\State {Remove $w$ from $F$}
\EndIf
\State {Pick any leaf $x$ of $F$}
\State {Color $vx$ by $\alpha$}
\State {Remove $x$ from $F$}
\EndProcedure
\end{algorithmic}

\textproc{Activate-u-Fan} takes time proportional to the length of the flipped path.
It always colors an edge and $F$ remains a (possibly degenerate) u-fan.
Since the colored edge is missing $\alpha$ at both endpoints, the partial edge-coloring of the graph remains legal.

%
%

\subsection*{Outline of \textproc{Color-Many}}

\textproc{Color-Many} is based on a set of disjoint fans that use the same color $\alpha$ as their first parameter.
We refer to such a set as an \emph{$\alpha$-collection}.
Formally, for a color $\alpha \in [d+1]$:
\begin{definition}
An \emph{$\alpha$-collection} is a set $\cC_\alpha$ of vertex-disjoint fans such that
every fan $F \in \cC_\alpha$ is either a primed c-fan $(\alpha, v, x_0, x_1, \dots, x_k)$ or a u-fan $(\alpha, v, x_1, \dots, x_k)$.
Let $V(\cC_\alpha)$ denote the set of vertices in all the fans of $\cC_\alpha$.
\end{definition}

\textproc{Color-Many} builds an $\alpha$-collection $\cC_\alpha$ around a set of uncolored edges where each edge has one endpoint missing $\alpha$. Then it uses $\cC_\alpha$ to color a large fraction of those edges.
It proceeds in three steps, each requiring $O(m)$ time, that work broadly as follows.
\bi[leftmargin=*]
\item[]\textbf{Step 1: Choose $\alpha$}

Say a vertex is \emph{incomplete} if it is incident to an uncolored edge.
Let $I$ be the set of incomplete vertices.
For any color $\beta$, let $I_\beta = \{v \in I \mid \beta \in M(v)\}$.
Choose $\alpha$ to maximize $|I_\alpha|$.

\item[]\textbf{Step 2: Build an $\alpha$-collection}

Build an $\alpha$-collection $\cC_\alpha$ using the vertices of $I_\alpha$.
The goal is for many vertices of $I_\alpha$ to end up in $V(\cC_\alpha)$ as the centers of primed c-fans or the leaves of u-fans.
\textproc{Color-Many} does this by a procedure called \textproc{Build-Collection}.

\item[]\textbf{Step 3: Activate the $\alpha$-collection}

Activate each fan in $\cC_\alpha$ using \textproc{Activate-c-Fan} or \textproc{Activate-u-Fan}, as appropriate.
After an activation, find any other fans in the collection that have been damaged and either repair them or discard them so that $\cC_\alpha$ remains an $\alpha$-collection.
\textproc{Color-Many} does this by a procedure called \textproc{Activate-Collection}.
\ei
Note that edges may be colored in both step 2 and step 3.

\begin{algorithmic}[1]
\Procedure {Color-Many}{{}}
\State {Choose $\alpha$ that maximizes $|I_\alpha|$}
\State {$\cC_\alpha \gets \textproc{Build-Collection}(\alpha, I_\alpha)$}
\State {\textproc{Activate-Collection($\cC_\alpha$)}}
\EndProcedure
\end{algorithmic}

These steps are laid out in detail in the next sections. We shall prove the following theorem.
\begin{theorem}[\ref{thm:many}]
\textproc{Color-Many} colors $\Omega(\ell/d)$ edges in $O(m)$ time.
\end{theorem}

\paragraph*{Step 1: Choose $\alpha$}

This step is straightforward.
\begin{lemma}
\label{lem:Ialpha}
There is some color $\alpha$ such that $|I_\alpha| \geq 2\ell/(d+1)$, and $\alpha$ can be found in $O(m)$ time.
\end{lemma}
\begin{proof}
Observe that
\[\sum_{\beta \in [d+1]} |I_\beta| = \sum_{v\in I} |M(v)| \geq 2\ell\]
Choose $\alpha$ to be the color missing at the largest number of incomplete vertices.
Then $|I_\alpha| \geq 2\ell/(d+1)$.

We can find $\alpha$ by counting the number of edges of each color incident to an incomplete vertex (an edge is counted twice if both endpoints are incomplete) and choosing $\alpha$ to be the color with the lowest count. Finding $\alpha$ takes $O(m)$ time. Then $I_\alpha$ can be constructed in $O(m)$ time.
\end{proof}

\paragraph*{Step 2: Build an $\alpha$-Collection}
\textproc{Color-Many} calls a procedure \textproc{Build-Collection} to build an $\alpha$-collection $\cC_\alpha$ of disjoint fans containing vertices of $I_\alpha$.
This $\alpha$-collection shall have the property that each center of a c-fan and each leaf of a u-fan is a vertex of $I_\alpha$ (and the other vertices of $V(\cC_\alpha)$ are not in $I_\alpha$).

\textproc{Build-Collection} creates a set $S$ initially containing the vertices of $I_\alpha$.
While $S$ is nonempty, \textproc{Build-Collection} removes a vertex $v$ from $S$ and calls a procedure \textproc{Make-Collection-Fan}, which attempts to build a primed c-fan with $v$ as the center.

\begin{algorithmic}[1]
\Procedure {Build-Collection}{{$\alpha, I_\alpha$}}
\State {$\cC_\alpha \gets \emptyset$}
\State {$S \gets I_\alpha$}
\For {$v \in S$}
\State {Remove $v$ from $S$}
\State {Choose uncolored incident edge $vx_0$}
\State {$F \gets \textproc{Make-Collection-Fan}(v, x_0, \alpha, \cC_\alpha)$}
\State {If $F \not= \Null$, add $F$ to $\cC_\alpha$}
\EndFor
\EndProcedure
\end{algorithmic}

\textproc{Make-Collection-Fan} attempts to build a c-fan centered at $v$ in the same way as \textproc{Make-Primed-Fan}, with two important differences:
First, if a leaf added to the fan already appears in some other fan of the collection, then the two intersecting fans need to be altered because fans must be disjoint in an $\alpha$-collection. This is done by a procedure called \textproc{Merge-Fans}.
Second, if a leaf $x$ is missing $\alpha$, then the fan is shifted from $x$ immediately and $vx$ is colored by $\alpha$.
This is necessary to ensure that a vertex of $I_\alpha$ does not end up as a leaf of a c-fan.


Formally, \textproc{Make-Collection-Fan} works as follows given an uncolored edge $vx_0$ with $v \in I_\alpha$.
Initialize a c-fan $F = (\alpha, v, x_0)$.
Repeatedly, letting $x_k$ be the last leaf of $F$, do the following:
\be
\item
Check whether $x_k$ already belongs to a different fan $F'$ in $\cC_\alpha$.
If it does, then $F$ and $F'$ overlap at $x_k$.
Call $\textproc{Merge-Fans}(F, F')$ and return the result.
\item
If $\alpha \in M(x_k)$ then shift $F$ from $x_k$.
Now $vx_k$ is uncolored and $\alpha$ is missing at both endpoints.
Color $vx_k$ by $\alpha$ and discard $F$.
If $x_k \in S$, then remove $x_k$ from $S$.
\item
Otherwise, proceed as in \textproc{Make-Primed-Fan}: Choose a color $\beta \in M(x_k)$.
If $\beta \in M(v)$, then prime $F$ with $\beta$ and return it.
If $\beta \not\in M(v)$, let $x_{k+1}$ be the vertex with $c(vx_{k+1}) = \beta$.
If $x_{k+1}$ already appears in $F$, then prime $F$ by $\beta$ and return $F$.
Otherwise append $x_{k+1}$ to $F$ as a leaf.
\ee

\textproc{Merge-Fans}$(F, F')$ accepts a c-fan $F$ whose last leaf $x_k$ also lies in the fan $F'$.
First, it shifts $F$ from $x_k$ so that $vx_k$ is uncolored.
Then it operates in four cases, depending on whether $F'$ is a c-fan or a u-fan and whether $x_k$ is the center of $F'$ or a leaf.
In all cases, the fan $F$ will be discarded (i.e. not included in the $\alpha$-collection).
These cases are illustrated in Figure~\ref{fig:mergecases}.

\be[label={\textbf{Case \Roman*:}}, leftmargin=\widthof{\textbf{Case III:}}+\labelsep]
\item If $F'$ is a c-fan and $x_k$ is its center, then $\alpha$ is missing at both $v$ and $x_k$.
Color $vx_k$ by $\alpha$ and discard $F'$.
\item If $F'$ is a c-fan (with center $u$) and $x_k$ is a leaf of $F'$, then create a u-fan as follows:
Shift $F'$ from $x_k$.
Now $ux_k$ and $vx_k$ are uncolored and $\alpha$ is missing at both $u$ and $v$.
By the construction of $F'$, $\alpha \not\in M(x_k)$.
Create a new u-fan having $x_k$ as its center and $u$ and $v$ as its leaves.
Add the u-fan to $\cC_\alpha$ and discard $F'$.

\item If $F'$ is a u-fan and $x_k$ is its center, then add $v$ as a leaf of $F'$.
\item If $F'$ is a u-fan and $x_k$ is a leaf of $F'$, then $\alpha$ is missing at both $v$ and $x_k$.
Color $vx_k$ by $\alpha$.
Remove $x_k$ from $F'$; if this causes $F'$ to be degenerate (i.e. have less than two leaves) then discard $F'$.
\ee

\begin{figure}[!ht]
\begin{center}
\includegraphics[scale=0.36]{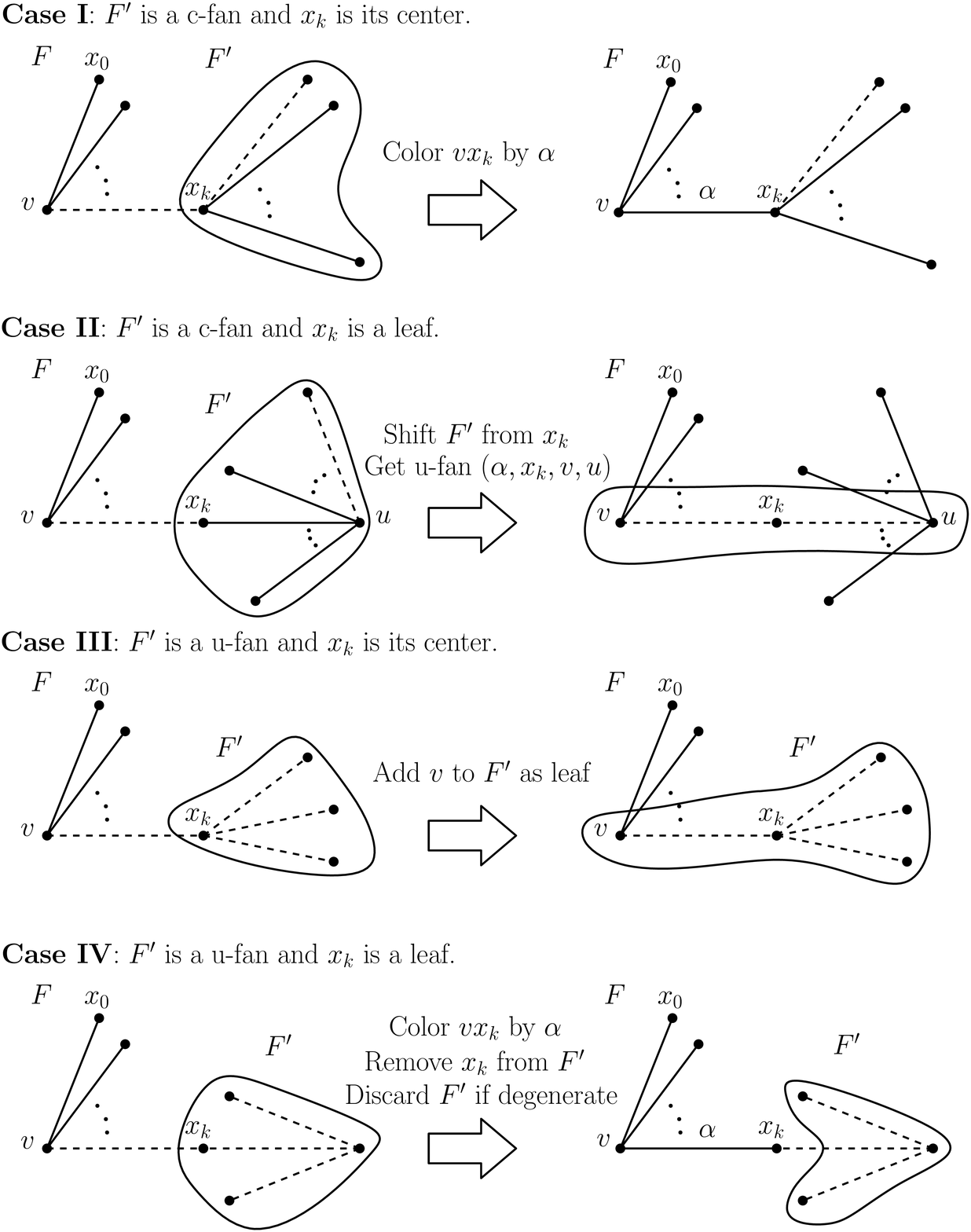}
\end{center}
\caption{A visual summary of \textproc{Merge-Fans}. In each of the four cases, a c-fan $F$ intersects at its last leaf $x_k$ with some fan $F'$ of $\cC_\alpha$. The first step is always to shift $F$ from $x_k$. The left side shows the configuration of the intersection (after the shift) and the right side shows the result of \textproc{Merge-Fans}. The circled vertices in each figure indicate the fans that belongs to $\cC_\alpha$.}
\label{fig:mergecases}
\end{figure}

\paragraph*{Analysis of Build-Collection}

Let us check the correctness and runtime of \textproc{Build-Collection}.

\begin{lemma}
\label{lem:buildTime}
$\textproc{Build-Collection}(\alpha, I_\alpha)$ makes an $\alpha$-collection in $O(m)$ time.
\end{lemma}
\begin{proof}
\textproc{Build-Collection} maintains the invariant that every vertex of $V(\cC_\alpha) \cap I_\alpha$ is either the center of a c-fan or a leaf of a u-fan.
Moreover, $\alpha$ cannot be missing at any leaf of a c-fan in $\cC_\alpha$.
These are ensured by steps 1 and 2 of \textproc{Make-Collection-Fan}.

Let us check that $\cC_\alpha$ remains an $\alpha$-collection after a call to \textproc{Make-Collection-Fan}.
It can terminate in steps 1, 2, or 3.

If it terminates during step 3, then $F$ is a primed c-fan created just as in \textproc{Make-Primed-Fan} and it is disjoint from every other fan in the collection.
If it terminates because of step 2, then a new edge is colored and $F$ is discarded.
No other fans are affected in these cases, and hence $\cC_\alpha$ is an $\alpha$-collection.

If it terminates because of step 1, then $F$ is merged with a fan $F'$ in the collection because they overlap at $x_k$. \textproc{Merge-Fans} only touches the vertices of $F$ and $F'$, so no fans are affected except for these two. It either returns a new u-fan to be added to the collection (Case II), adds a new leaf to an existing u-fan (Case III), or colors a new edge (Cases I and IV), leaving $F'$ in the collection only if it is still a well-formed fan.
By inspection of the cases, every fan added to $\cC_\alpha$ is well-formed and disjoint from all others in $\cC_\alpha$.
Thus, \textproc{Build-Collection} outputs an $\alpha$-collection.

For the runtime, observe that only the vertices of $I_\alpha$ can ever be the centers of c-fans and at most one c-fan is built around each $v \in I_\alpha$ during \textproc{Build-Collection}.
That is, once a vertex of $I_\alpha$ \emph{stops} being the center of a c-fan, it is never again the center of a c-fan.
It takes $O(\deg(v))$ time to build a c-fan around a center $v \in I_\alpha$.
Hence the total time spent building c-fans is $O\left(\sum_{v \in I_\alpha} \deg(v)\right) = O(m)$.
This accounts for all time spent in steps 2 and 3 of \textproc{Make-Collection-Fan}.

The remaining time is spent on \textproc{Merge-Fans}.
Every operation in \textproc{Merge-Fans} takes constant time except for shifting $F$ and, in Case II, shifting $F'$.
In all cases, any shifted fan is always discarded immediately afterwards, and so each c-fan that is built can be shifted only once.
The total time for all shifts is therefore bounded by the sum of sizes of all c-fans: $O\left(\sum_{v \in I_\alpha} \deg(v)\right) = O(m)$.
Excluding the time spent on shifts, \textproc{Merge-Fans} takes constant time each time it is called (and it is called at most $I_\alpha|$ times).
Hence \textproc{Build-Collection} takes $O(m)$ time.
\end{proof}

Now let us estimate the number of edges colored during \textproc{Build-Collection} and the number of fans in $\cC_\alpha$.
There is be a balance between these two quantities because some vertices of $I_\alpha$ do not end up in $V(\cC_\alpha)$, but that only happens in cases where some edge is colored.

\begin{lemma}
\label{lem:buildNumber}
Let $\ell_b$ and $\ell_m$, respectively, be the number of uncolored edges before and after \textproc{Build-Collection}.
Let $\cC_\alpha$ be the collection output by \textproc{Build-Collection}$(I_\alpha)$.
Then 
\[|V(\cC_\alpha) \cap I_\alpha| + 3(\ell_b - \ell_m) \geq |I_\alpha|\]
\end{lemma}
\begin{proof}
Let $v \in I_\alpha$.
There are only a few ways that \textproc{Build-Collection} can affect $v$.

Possibly, $v$ could end up in some fan of $\cC_\alpha$, either as the center of a c-fan or a leaf of a u-fan.
That is, $v$ could be in $V(\cC_\alpha) \cap I_\alpha$.

If this does not happen, it is because $v$ is part of a discarded fan or removed from its fan.
These events can occur during the construction of a c-fan (in step 2 of \textproc{Make-Collection-Fan}) or in Cases I and IV of \textproc{Merge-Fans}.
However, all three of these scenarios coincide with an uncolored edge being colored.
Observe that every time an edge is colored, at most three\footnote{The worst case happens in Case IV when $F'$ is a u-fan with two leaves, one being $x_k$. The edge $vx_k$ is colored and $v$ and the two leaves of $F'$ are discarded.} vertices are removed from $V(\cC_\alpha) \cap I_\alpha$.

Therefore, at most $3(\ell_b - \ell_m)$ vertices of $I_\alpha$ do not end up in a fan of $\cC_\alpha$.
It follows that
\[|V(\cC_\alpha) \cap I_\alpha| + 3(\ell_b - \ell_m) \geq |I_\alpha|\qedhere\]
\end{proof}

\paragraph*{Step 3: Activate the $\alpha$-Collection}

Once $\cC_\alpha$ is constructed, \textproc{Color-Many} calls \textproc{Activate-Collection} to color a large fraction of the uncolored edges in the fans of $\cC_\alpha$.
Together with the edges colored in step 2, $\Omega(\ell/d)$ edges will be colored in total.

\textproc{Activate-Collection} works in $d$ stages, one for each color other than $\alpha$.
Let $\beta_1, \dots, \beta_d$ be the colors of $[d+1]\setminus \{\alpha\}$ in some fixed order.
For $i = 1, \dots, d$, let 
\begin{align*}
Q_i = \{F \in \cC_\alpha \mid &\text{ $F$ is a c-fan primed by $\beta_i$ or}\\
&\text{ $F$ is a u-fan and $\beta_i$ is missing at its center}\}
\end{align*}

We assume that the algorithm correctly maintains all the sets $Q_i$ at all times.
This can be handled automatically with constant overhead when coloring and uncoloring edges.

\textproc{Activate-Collection} runs in $d$ stages, where the $i$th stage deals only with the fans of $Q_i$.
Stage $i$ works as follows.
While $Q_i$ is nonempty, pick a fan $F \in Q_i$ and let $v$ be the center of $F$.
If $F$ is a c-fan (primed by $\beta_i$) then call $\textproc{Activate-c-Fan}(F)$.
If $F$ is a u-fan (with $\beta_i \in M(v)$) then call $\textproc{Activate-u-Fan}(F, \beta_i)$.
After the activation, discard $F$ if it is a c-fan or a degenerate u-fan.
Then call \textproc{Disconnect-Vertex}$(v)$, a procedure that cleans up the damage that may have been caused by the activation.

\textproc{Activate-Collection} is given in pseudocode below.
\begin{algorithmic}[1]
\Procedure {Activate-Collection}{{$Q_1, \dots, Q_d$}}
\For {$i$ from $1$ to $d$}
\State\AlgComment {Stage $i$}
\While {$Q_i \not= \emptyset$}
\State {Pick $F \in Q_i$}
\If {$F = (\alpha, v, x_0, \dots, x_k)$ is a c-fan primed by $\beta_i$}
\State {$\textproc{Activate-c-Fan}(F)$}
\State {Remove $F$ from $\cC_\alpha$}
\ElseComment {$F = (\alpha, v, x_1, \dots, x_k)$ is a u-fan and $\beta_i \in M(v)$}
\State {$\textproc{Activate-u-Fan}(F, \beta_i)$}
\If {$F$ has less than two leaves}
\State {Remove $F$ from $\cC_\alpha$}
\EndIf
\EndIf
\State {\textproc{Disconnect-Vertex}$(v)$}
\EndWhile
\EndFor
\EndProcedure
\end{algorithmic}

Except for the call to \textproc{Disconnect-Vertex}, this procedure doesn't do anything surprising.
It activates each fan using the appropriate function, and then discards it unless it is a u-fan that could be activated again later.
The purpose of \textproc{Disconnect-Vertex} is to address two issues.

First, if the alternating path that was flipped during the activation ends at a vertex $w$ of some fan $F'$, then $F'$ may no longer be well-formed (i.e. may not adhere to the definition of a u-fan or a primed c-fan) because $M(w)$ has changed --- specifically, $\alpha$ and $\beta_i$ have been exchanged in $M(w)$.
Whenever a fan $F'$ is malformed for this reason, we say $F'$ is \emph{damaged} and that $w$ is the \emph{damaged vertex} of $F'$.

A damaged fan must be repaired or discarded before the next activation, for otherwise $\cC_\alpha$ contains some ``fan-like object'' that does not satisfy the definition of a fan. Let us formalize this requirement.

\be[leftmargin=\widthof{\textbf{Precondition 1:}}+\labelsep]
\item[\textbf{Precondition 1:}] $\cC_\alpha$ is always an $\alpha$-collection before a fan is activated during \textproc{Activate-Collection}.
In particular, no fan in $\cC_\alpha$ is damaged.
\ee

Second, the edges of the flipped path may now belong to some larger $\alpha\beta_i$-path in the graph which could be flipped later during a different activation. 
If the same edges are in many different flipped paths during stage $i$, the runtime could exceed $O(m)$.
We avoid this by isolating every vertex in the flipped path so that it cannot be accessed again during stage $i$.

During stage $i$, say a vertex $x$ is \emph{disconnected} if it lies in an $\alpha\beta_i$-path whose endpoints are not in $V(\cC_\alpha)$.

\be[leftmargin=\widthof{\textbf{Precondition 2:}}+\labelsep]
\item[\textbf{Precondition 2:}]
Before a fan is activated during stage $i$, any vertex that had previously been in a flipped path during stage $i$ is disconnected.
\ee

Note that \textproc{Disconnect-Vertex} may also flip $\alpha\beta_i$-paths, and the vertices in those paths must satisfy Precondition 2 as well.

\paragraph*{Disconnect-Vertex}

The purpose of \textproc{Disconnect-Vertex} is to ensure that these preconditions are satisfied at the top of each iteration.
\textproc{Disconnect-Vertex} is given a vertex $v$ and it explores the $\alpha\beta_i$-path $P$ containing it.
It manipulates the fans containing the endpoints of the path with the twofold goal of repairing those fans if they are damaged (with the endpoint being the damaged vertex) and ensuring that the endpoints of $P$ do not lie in fans of $\cC_\alpha$.


The procedure works as follows.
Let $P$ be the $\alpha\beta_i$-path containing $v$.
Until neither endpoint of $P$ lies in a fan of $\cC_\alpha$, do the following:
Choose $w$ to be an endpoint of $P$ that lies in a fan $F' \in \cC_\alpha$, and if possible choose $w$ so that $F'$ is damaged (with $w$ being the damaged vertex).
Then act in three cases.
\be[label={\textbf{Case \arabic*:}}, leftmargin=\widthof{\textbf{Case 3:}}+\labelsep]
\item If $F'$ is a c-fan or a u-fan with only two leaves, then discard $F'$ (i.e. delete it from $\cC_\alpha$).
\item If $F'$ is a u-fan with at least three leaves and $w$ is a leaf of $F'$, then remove $w$ from $F'$.
\item If $F'$ is a u-fan with at least three leaves and $w$ is the center of $F'$, then choose a leaf $x$ of $F'$.
If $x \in P$, replace $x$ with a different leaf of $F'$.

Now we have two subcases depending on whether $\alpha \in M(w)$.

\be[label={\textbf{Case 3\alph*:}}, leftmargin=\widthof{\textbf{C:}}+\labelsep]
\item
If $\alpha \in M(w)$, then $w$ must be the damaged vertex of $F'$.
We have $\alpha \in M(w) \cap M(x)$.
Color $wx$ by $\alpha$ and remove $x$ from $F'$.
Now $F'$ is undamaged.
\item
If $\alpha \not\in M(w)$, then $F'$ is a valid u-fan.
We must have $\beta_i \in M(w)$ for otherwise $w$ would not be an endpoint of $P$.
Flip the $\alpha\beta_i$-path beginning at $x$.
Now $\beta_i \in M(w) \cap M(x)$.
Color $wx$ by $\beta_i$ and remove $x$ from $F'$.
\ee
Now $wx$ is colored by $\alpha$ or $\beta_i$, and so $w$ is no longer an endpoint of $P$.
Extend $P$ by continuing the $\alpha\beta_i$-path through $wx$ until it reaches an endpoint.
\ee
\textproc{Disconnect-Vertex} is given below in pseudocode.
\begin{algorithmic}[1]
\Procedure {Disconnect-Vertex}{{$v$}}
\State {Let $P$ be the $\alpha\beta_i$-path containing $v$}
\While {an endpoint of $P$ lies in a fan of $\cC_\alpha$}
\State {Let $e_1$ and $e_2$ be the endpoints of $P$}
\If {$e_1$ is the damaged vertex of a fan or $e_2$ is not in a fan}
\State {$w \gets e_1$}
\Else
\State {$w \gets e_2$}
\EndIf
\State {Let $F'$ be the fan containing $w$}
\Statex {\textbf{Case 1}}
\If {$F'$ is a c-fan or a u-fan with two leaves}
\State {Remove $F'$ from $\cC_\alpha$}
\Else
\Statex {\textbf{Case 2}}
\If {$F'$ is a u-fan and $w$ is a leaf of $F'$}
\State {Remove $w$ from $F'$}
\Else
\Statex {\textbf{Case 3}}
\If {$F'$ is a u-fan and $w$ is the center of $F'$}
\State {Choose a leaf $x$ of $F'$}
\If {$x \in P$}
\State {Replace $x$ with a different leaf of $F'$}
\EndIf
\State {Remove $x$ from $F'$}
\If {$\alpha \in M(w)$}
\Statex {\qquad\qquad\quad\textbf{Case 3a}}
\State\AlgComment {$F'$ is damaged as $\alpha \in M(w)$}
\State\AlgComment {$\alpha \in M(w) \cap M(x)$}
\State {Color $wx$ by $\alpha$}
\Else
\Statex {\qquad\qquad\quad\textbf{Case 3b}}
\State\AlgComment {$F'$ is not damaged}
\State {Flip the $\alpha\beta_i$-path containing $x$}
\State\AlgComment {$\beta_i \in M(w) \cap M(x)$}
\State {Color $wx$ by $\beta_i$}
\EndIf
\State\AlgComment{$w$ is no longer an endpoint of $P$}
\State {Extend $P$ by following the $\alpha\beta_i$-path through $wx$}
\EndIf
\EndIf
\EndIf
\EndWhile
\EndProcedure
\end{algorithmic}

\paragraph*{Analysis of \textproc{Disconnect-Vertex}}

Let us first establish a few facts about \textproc{Disconnect-Vertex}.

\begin{proposition}
\label{prop:disconnect}
Consider a fan $F$ with center $v$ that is activated during stage $i$ of \textproc{Activate-Collection}.
Let $P_F$ be the $\alpha\beta_i$-path flipped during the activation and let $v$ and $w$ be its endpoints, and assume Preconditions 1 and 2 held before the activation of $F$.
At all times during \textproc{Disconnect-Vertex}$(v)$, the following hold.
\be
\item
$P$ contains $P_F$.
\item
$P$ is a maximal $\alpha\beta_i$-path. In particular, $P$ never becomes a cycle.
\item
At most one endpoint of $P$ is the damaged vertex of a fan in $\cC_\alpha$ and no other fan in $\cC_\alpha$ is damaged.
\ee
\end{proposition}
\begin{proof}
We check these in order.
\be
\item
At the start of \textproc{Disconnect-Vertex}$(v)$, $P$ contains $P_F$.
This can be checked by case analysis of \textproc{Activate-c-Fan} and \textproc{Activate-u-Fan}.
Since vertices are never removed from $P$, it remains true.
\item
Initially, $P$ is chosen to be the (maximal) $\alpha\beta_i$-path containing $v$.
Suppose $P$ is a maximal $\alpha\beta_i$-path at the start of the loop (at line 4 in \textproc{Disconnect-Vertex}).
At the end of the loop, either $P$ will be unchanged and will still be maximal (as in Cases 1 and 2), or $P$ will have been extended at one endpoint by Case 3.
In all cases, $P$ will be maximal.

We must check that $P$ cannot become a cycle due to Case 3.
Case 3 extends $P$ through an edge $wx$ of a u-fan $F'$, where $w$ is the center of $F'$ (and an endpoint of $P$) and $x$ is a leaf of $F'$.
Observe that any leaf of $F'$ that lies in $P$ must be the other endpoint of $P$, since $\alpha$ is missing at each leaf.
If necessary, the procedure replaces $x$ by a different leaf so that $x$ is not in $P$.
Thus, when $wx$ is colored to join $P$ with the $\alpha\beta_i$-path containing $x$, it cannot create a cycle.

\item
Before \textproc{Disconnect-Vertex} is called, only the fan containing $w$ (if any) can be damaged and only with $w$ being the damaged vertex.
Indeed, $M(w)$ has been altered because $P_F$ was flipped, and so that fan may not satisfy the definition of a u-fan or primed c-fan.
We check that all other fans in $\cC_\alpha$ are still well-formed.
The activation of $F$ only affected fans containing vertices of $P_F$ (including $F$ itself).
The fans containing only \emph{interior} vertices of $P_F$ cannot be made invalid due to the activation of $F$ because flipping $P_F$ does not change the set of missing colors at any interior vertex.
The fan $F$ is discarded immediately unless it is a non-degenerate u-fan after its activation, in which case it is still well-formed.
Thus, only the fan containing $w$ can be malformed, and only in the sense that $w$ is the damaged vertex of that fan.

Now assume inductively that at most one endpoint of $P$ is the damaged vertex of a fan in $\cC_\alpha$ and all other fans are well-formed.
If an endpoint of $P$ is the damaged vertex of a fan $F'$, then algorithm chooses $w$ to be that vertex.
In Case 1, $F'$ is discarded and so all fans left in $\cC_\alpha$ are well-formed.
In Case 2, $w$ is removed from $F'$ which repairs $F'$, and $w$ is no longer in any fan of $\cC_\alpha$.
In Case 3a, $F'$ is damaged initially and repaired by coloring an incident edge by $\alpha$; at this point, no fan can be damaged.
After these cases, no fan in $\cC_\alpha$ can be damaged.

In Case 3b, $F'$ is not damaged and so no fan in $\cC_\alpha$ is damaged, for otherwise $w$ would have been chosen to be the damaged vertex.
However, when the $\alpha\beta_i$-path containing $x$ is flipped, the fan on the other end could be damaged.
The damaged vertex of this fan becomes the new endpoint of $P$ when $P$ is extended through the edge $wx$.
Thus, there is at most one damaged fan at the end of this case and its damaged vertex must be an endpoint of $P$.
\qedhere
\ee
\end{proof}

Now we can check that \textproc{Disconnect-Vertex} ensures Preconditions 1 and 2. Establishing Precondition 1 will imply the correctness of \textproc{Activate-Collection}, and Precondition 2 will be useful in measuring its runtime.

\begin{lemma}
Preconditions 1 and 2 hold before each fan is activated in \textproc{Activate-Collection}.
\end{lemma}
\begin{proof}

Both preconditions hold before any fan is activated.
Assume inductively that the preconditions hold before the activation of a fan $F$.

By Proposition~\ref{prop:disconnect}, no fan can be damaged or otherwise malformed except for those containing the endpoints of $P$, and neither endpoint of $P$ lies in a fan of $\cC_\alpha$ when \textproc{Disconnect-Vertex}$(v)$ terminates.
Clearly \textproc{Disconnect-Vertex} must terminate as each iteration decreases $|V(\cC_\alpha)|$. 
Therefore Precondition 1 is restored.

Now we check Precondition 2.
By Proposition~\ref{prop:disconnect}, $P$ is always maximal $\alpha\beta_i$-path and it contains all the vertices of the path flipped during the activation of $F$.
When the procedure is over, every vertex in $P$ is disconnected because the endpoints of $P$ are not in fans of $\cC_\alpha$.

All other vertices that were disconnected before $F$ was activated must still be disconnected.
This holds because every edge colored by $\alpha$ or $\beta_i$ that is affected by the activation of $F$ and the subsequent call to \textproc{Disconnect-Vertex} becomes an edge in $P$.
Hence Precondition 2 is restored.
\end{proof}

\begin{lemma}
\label{lem:disjoint}
The paths $P_1, \dots, P_k$ traversed by the applications of \textproc{Disconnect-Vertex} during stage $i$ are vertex-disjoint.
\end{lemma}
\begin{proof}
Proceed by induction on $k$.
Let $P_k$ be the $\alpha\beta_i$-path constructed by \textproc{Disconnect-Vertex} after the activation of a fan $F$ with center $v$.
By Precondition 2, all the vertices in $P_1, \dots, P_{k-1}$ are disconnected before the activation.
Since $v$ was the center of $F$, $v$ was not disconnected before the activation, and so it cannot lie in $P_1, \dots, P_{k-1}$.
At the start of \textproc{Disconnect-Vertex}$(v)$, $P_k$ is the $\alpha\beta_i$-path containing $v$.
Whenever \textproc{Disconnect-Vertex} extends $P_k$ by coloring an edge (as in Case 3), the $\alpha\beta_i$-path that forms the extension has an endpoint in a fan of $\cC_\alpha$.
Hence a path $P_j$ ($j < k$) cannot be traversed during the construction of $P_k$, and so $P_k$ is disjoint from all previous paths.
It follows that $P_1, \dots, P_k$ are disjoint.
\end{proof}

Finally, we measure the runtime of \textproc{Disconnect-Vertex}.

\begin{lemma}
\label{lem:disconnectTime}
\textproc{Disconnect-Vertex} takes $O(\length(P))$ time, where $P$ is the constructed path.
\end{lemma}
\begin{proof}
The runtime is clear as constant time is spent on each vertex of $P$ and vertices are never removed from $P$.
\end{proof}

\paragraph*{Analysis of \textproc{Activate-Collection}}

We aim to prove that \textproc{Activate-Collection} runs in $O(m)$ time and colors a number of edges that is at least a constant fraction of $|V(\cC_\alpha) \cap I_\alpha|$.

\begin{lemma}
\label{lem:activateTime}
\textproc{Activate-Collection} runs in $O(m)$ time.
\end{lemma}
\begin{proof}
Consider a fan $F$ in $\cC_\alpha$.
If $F$ is a c-fan, then it takes $O(|F| + \length(P))$ time to activate $F$ and call \textproc{Disconnect-Vertex}, where $P$ is the entire $\alpha\beta_i$-path formed by \textproc{Disconnect-Vertex}.
This follows from Lemma~\ref{lem:disconnectTime} and the fact that the path flipped during the activation of $F$ is contained in $P$.
If $F$ is a u-fan, then the time is just $O(\length(P))$ because the activation itself (not including the flipping operation) takes constant time.

Let $P_1, \dots, P_T$ be the paths explored by \textproc{Disconnect-Vertex} during \textproc{Activate-Collection}.
Choose indices $1 = T_0 \leq T_1 \leq \cdots \leq T_d = T + 1$ such that $P_{T_{i-1}}, \dots, P_{T_i - 1}$ are the paths constructed during stage $i$.
The total time is then 
\[O\left(\sum_{\text{c-fan }F\in\cC_\alpha} |F| + \sum_{i=1}^d\sum_{t=T_{i-1}}^{T_i-1} \length(P_t)\right)\]
We have $\sum_{\text{c-fan }F \in \cC_\alpha} |F| \leq m$ as the c-fans are disjoint and a c-fan can only be activated once.
We now show that the sum of the lengths of all paths $P_t$ is $O(m)$.

For $i = 1, \dots, d$, let $m_i$ be the number of edges colored by $\beta_i$ at the end of stage $i$.
Observe that the $i$th stage only modifies the number of edges colored by $\alpha$ and $\beta_i$ --- all others are fixed.
Indeed, the only operation in stage $i$ that involves an edge not colored by $\alpha$ or $\beta$ is the shift operation, and shifting a c-fan does not change the number of edges of each color.
That means, for each $i$, the number of edges colored by $\beta_i$ after stage $d$ is also $m_i$.
Therefore $\sum_{i=1}^d m_i \leq m$.

Now consider stage $i$.
By Lemma~\ref{lem:disjoint}, $P_{T_{i-1}}, \dots, P_{T_i - 1}$ are disjoint.
At the end of stage $i$, each path $P_t$ that was built during stage $i$ is an $\alpha\beta_i$-path, and so has at most one more edge colored by $\alpha$ than by $\beta_i$.
That is, if $P_t$ has $k_t$ edges colored by $\beta_i$ then $\length(P_t) \leq 2k_t + 1$.
Thus, \[\sum_{t=T_{i-1}}^{T_i-1} \length(P_t) \leq \sum_{t=T_{i-1}}^{T_i-1} (2k_t + 1) \leq 2m_i + (T_i - T_{i-1})\]

Summing this estimate over all rounds gives us
\[\sum_{i=1}^d\sum_{t=T_{i-1}}^{T_i-1} \length(P_t) \leq \sum_{i=1}^d 2m_i + (T_i - T_{i-1})
= 2\left(\sum_{i=1}^d m_i\right) + (T_d - T_0)
\leq 2 m + T
\leq 3m
\]
The last inequality holds because every activation colors a new edge, and so $T \leq m$.
Therefore, \textproc{Activate-Collection} runs in $O(m)$ time.
\end{proof}

Now we establish that \textproc{Activate-Collection} does not ignore any fans in $\cC_\alpha$.

\begin{lemma}
\label{lem:empties}
$\cC_\alpha$ is empty after \textproc{Activate-Collection}$(\cC_\alpha)$.
\end{lemma}
\begin{proof}
\textproc{Activate-Collection} runs in $d$ stages, where the $i$th stage finishes once $Q_i$ is empty.
We show that no fan can be added to $Q_i$ in later stages.
A c-fan cannot be added to $Q_i$ in later stages because the color that primes a c-fan never changes.
A u-fan $F \not\in Q_i$ can only be added to $Q_i$ when its center is the endpoint of a flipped path that interchanges $\alpha$ and $\beta_i$.
However, $\alpha\beta_i$-paths are only flipped during stage $i$, and so u-fans can only be added to $Q_i$ during stage $i$.
Thus, $Q_i$ remains empty after stage $i$.
It follows that $Q_1, \dots, Q_d$ are all empty at the end of \textproc{Activate-Collection}.
\end{proof}

Next, we compare the number of edges colored by \textproc{Activate-Collection} against the number of vertices removed from $V(\cC_\alpha) \cap I_\alpha$.

\begin{lemma}
\label{lem:activateNumber}
In every iteration of \textproc{Activate-Collection}, the number of uncolored edges decreases by $r \geq 1$ and at most $r+6$ vertices of $I_\alpha$ are removed from $V(C_\alpha)$.
\end{lemma}
\begin{proof}
Recall that the center of a c-fan is in $I_\alpha$ (but none of the leaves are) and each leaf of a u-fan is in $I_\alpha$ (but the center is not).
Activating a fan always causes a new edge to be colored and it removes a vertex of $I_\alpha$ from that fan: either the center of a c-fan or a leaf of a u-fan.
If it is a u-fan then another vertex of $I_\alpha$ may be removed if the flipped path ends at a leaf of that fan, and a third vertex of $I_\alpha$ may be removed if the u-fan becomes degenerate and is discarded.
Thus, due to the activation, one new edge is colored and $|V(\cC_\alpha) \cap I_\alpha|$ decreases by at most 3.

\textproc{Activate-Collection} then calls \textproc{Disconnect-Vertex}, which may color more edges and remove more vertices of $I_\alpha$.
In the first two cases of \textproc{Disconnect-Vertex}, no edges are colored and at most two vertices of $I_\alpha$ are removed from $V(\cC_\alpha)$.
These cases can happen only twice during \textproc{Disconnect-Vertex} because they cause an endpoint of $P$ to not lie in any fan of $\cC_\alpha$.
Hence, these cases cause $|V(\cC_\alpha) \cap I_\alpha|$ to decrease by at most 4.
Case 3 can occur many times, but it always colors one edge (either by $\alpha$ or by $\beta_i$) and removes one vertex of $V(\cC_\alpha) \cap I_\alpha$.

Thus, if $r$ new edges are colored during an iteration of \textproc{Activate-Collection}, at most $r + 6$ vertices are removed from $V(\cC_\alpha) \cap I_\alpha$: Up $3$ vertices are lost due the activation of the fan (while $1$ new edge is colored) and up to $(r-1) + 4$ vertices are lost during \textproc{Disconnect-Vertex} (while $r-1$ new edges are colored).
\end{proof}

The next corollary counts the number of edges newly colored by \textproc{Activate-Collection}$(\cC_\alpha)$.
\begin{corollary}
\label{cor:activateNumber}
Let $\ell_m$ and $\ell_f$, respectively, be the number of uncolored edges before and after \textproc{Activate-Collection}.
Let $\cC_\alpha$ be the collection output by \textproc{Build-Collection}.
Then
\[\ell_m-\ell_f \geq \frac{|V(\cC_\alpha) \cap I_\alpha|}{7}\]
\end{corollary}
\begin{proof}
By Lemma~\ref{lem:empties}, $V(\cC_\alpha) \cap I_\alpha$ is empty at the end of \textproc{Activate-Collection}. By Lemma~\ref{lem:activateNumber}, at least one edge is colored for every 7 vertices removed from $V(\cC_\alpha) \cap I_\alpha$.
\end{proof}

We are now able to wrap up the proof of Theorem~\ref{thm:many}.
\begin{theorem}
\label{thm:many}
\textproc{Color-Many} colors $\Omega(\ell/d)$ edges in $O(m)$ time.
\end{theorem}
\begin{proof}
\textproc{Color-Many} runs in $O(m)$ time by Lemmas~\ref{lem:Ialpha},~\ref{lem:buildTime}, and~\ref{lem:activateTime}.

Let $\ell_b$, $\ell_m$, and $\ell_f$, respectively, be the number of uncolored edges at the start of \textproc{Color-Many}, after \textproc{Build-Collection}, and after \textproc{Activate-Collection}.
Recall that $|I_\alpha| \geq 2\ell_b/(d+1)$ by Lemma~\ref{lem:Ialpha}.
By Lemma~\ref{lem:buildNumber},
\[|V(\cC_\alpha) \cap I_\alpha| + 3(\ell_b - \ell_m) \geq |I_\alpha|\]
By Corollary~\ref{cor:activateNumber},
\[\ell_m-\ell_f \geq \frac{|V(\cC_\alpha) \cap I_\alpha|}{7}\]
Hence,
\begin{align*}
\ell_b - \ell_f &= (\ell_b - \ell_m) + (\ell_m - \ell_f)\\
&\geq \frac{|I_\alpha| - |V(\cC_\alpha) \cap I_\alpha|}{3} + \frac{|V(\cC_\alpha) \cap I_\alpha|}{7}\\
&= \frac{7|I_\alpha| - 4|V(\cC_\alpha) \cap I_\alpha|}{21}\\
&\geq \frac{7|I_\alpha| - 4|I_\alpha|}{21}\\
&= \frac{|I_\alpha|}{7}\\
&\geq \frac{2\ell_b}{7(d+1)}
\end{align*}

Therefore, $\Omega(\ell/d)$ edges are colored by \textproc{Color-Many}.
\end{proof}

\bibliographystyle{plain}
\bibliography{EC}
\newpage
\appendix

\section{Dictionary Data Structure}
\label{app:dictionary}
Here we describe the implementation of $\cD$, the dictionary data structure introduced in Section~\ref{sec:datastructs}.
Recall that $\cD$ is meant to map a vertex-color pair $(v, \gamma)$ to the edge incident to $v$ colored by $\gamma$, if such an edge exists, and we wish to support \textproc{Search}, \textproc{Insert}, and \textproc{Delete} in constant time.

For this description, treat the vertices of $G$ as integers from $0$ to $n-1$.
Intuitively, we would like to simply store an array of length $n(d+1)$ where the entry at index $(v(d+1) + \gamma)$ contains the edge incident to $v$ colored by $\gamma$.
This may be too costly in preprocessing as $n(d+1)$ could exceed our desired time complexity of $O(m\sqrt n)$.
Instead we use a two-level decomposition of this method which takes $O(\sqrt{mnd})$ time to initialize.

$\cD$ must store key-value pairs, where the keys are drawn from a universe $V \times [d+1]$ of size $U = n(d+1)$ and we know that the dictionary contains at most $M = 2m$ entries at any time (at most one for each endpoint).

Let $b = \lceil\sqrt{U/M}\rceil$.
Divide the universe into $\lceil U/b\rceil$ ranges $R_0, \dots, R_{\lceil U/b\rceil - 1}$ of size at most $b$ (specifically, the ranges are $[0, b-1]$, $[b, 2b-1]$, etc).
Initialize an array $A$ of length $\lceil U/b\rceil$, where the $k$th entry of $A$ corresponds to the range $R_k$.
For each index $k$ of $A$, store a counter $C[k]$ that tracks the number of entries in range $R_k$.

Then initialize $M$ arrays $B_1, B_2, \dots, B_M$ of length $b$.
Call these \emph{block arrays}.
An entry of a block array is either empty ($\Null$) or it has a pointer to an edge.

If $C[k] = 0$, then $A[k]$ is empty (it contains $\Null$).
Otherwise, $A[k]$ contains a pointer to a block array $B$, and for as long as $C[k] > 0$, $B$ will act as an array storing the entries in the range $R_k$.
At any time, $B$ can only be linked from one location in $A$ (but it can change locations over time).

Create a linked list $E$ which shall contain the empty (unused) block arrays.
Initially all block arrays are in $E$.

The operations are implemented as follows:
\bi
\item
\textproc{Search}$(v, \gamma)$:
Write the index $v(d+1) + \gamma$ as a sum $kb + j$ with $k$ and $j$ integers.\footnote{Specifically $k =\lfloor (v(d+1) + \gamma)/b \rfloor$ and $j = (v(d+1) + \gamma) \mod b$.}
If $A[k] = \Null$ then return $\Null$. If $A[k] = B$, then return $B[j]$.
\item
\textproc{Insert}$(v, \gamma, e)$: Compute $k$ and $j$ as above.
If $A[k] = B$ then set $B[j] = \gamma$.
If $A[k] = \Null$, then remove an empty block array $B$ from $E$, set $A[k] = B$, and set $B[j] = \gamma$.
Increment $C[k]$.

Observe that $E$ cannot be empty when a block array is removed from $E$, for there are $M$ block arrays and at most $M$ entries in the dictionary.
\item
\textproc{Delete}$(v, \gamma)$: Compute $k$ and $j$ as above.
If $A[k] = \Null$, then do nothing.
If $A[k] = B$ and $B[j] \not= \Null$, then set $B[j] = \Null$ and decrement $C[k]$.
If $C[k] = 0$ now, then set $A[k] = \Null$ and append $B$ to $E$ because $B$ is empty.
\ei

Each operation takes constant time.
Initializing the data structure takes $O(U/b + Mb) = O(\sqrt{UM}) = O(\sqrt{mnd})$ time.

\paragraph*{A Note About Alternative Implementations}

Our implementation is somewhat unconventional, and it would be nice if an out-of-the-box dictionary could be used instead.
Our goal is to have constant time operations, and the barrier in our case is just the preprocessing time.
We need $\cD$ to be initialized in $O(m\sqrt{n})$ time.
There are alternative solutions available.

Firstly, if $m\sqrt{n} = \Omega(nd)$, as is often the case, we can simply use an array of length $n(d+1)$ as our dictionary.

Secondly, there are \emph{randomized} dictionaries that can be initialized in constant time and perform all operations in constant expected amortized time (e.g. generalized cuckoo hashing~\cite{FPSS05}).
If deterministic computation is not a priority, and certainly in the case of \textproc{Random-Euler-Color}, this type of dictionary can be substituted.

In fact, $\cD$ is only \emph{truly} needed for \textproc{Random-Color-One}.
Versions of \textproc{Color-One} and \textproc{Color-Many} exist that do not use $\cD$ at all; this was the case for \textproc{Recolor} and \textproc{Parallel-Color}, the analogous subroutines in~\cite{GNKLT85}.
Although $\cD$ is a valuable tool that simplifies their implementations, it is possible to achieve the same deterministic $O(m\sqrt{n})$ bound for edge-coloring without it.

\paragraph*{The Recursive Strategy and $\cD$}
We must be careful in how we manage our data structures with respect to the recursive strategy.
Since $\cD$ takes more than $O(m)$ time to initialize, we cannot afford to reinitialize $\cD$ for each subgraph during \textproc{Euler-Color} or \textproc{Random-Euler-Color}.
If done improperly, different branches of the recursion could run into conflicts by reading and writing to the same locations in $\cD$, even though they are operation on edge-disjoint subgraphs.

A simple workaround is to populate and depopulate $\cD$ in every \stp{Repair} step.
Before the \stp{Repair} step begins, insert the color of every edge into $\cD$; when it is over, remove every color from $\cD$.
This way, every time a \stp{Repair} step occurs during the strategy, $\cD$ exactly reflects the current subgraph and its colors.
It adds $O(m)$ insertions and deletions to every \stp{Repair} step which does not change the asymptotic runtime.
It is also possible to do without this workaround by careful management of the colors.

\end{document}